\documentclass[lettersize,journal]{IEEEtran}
\pdfoutput=1

\IEEEoverridecommandlockouts
\usepackage{algorithm}
\usepackage{algpseudocode}

\renewcommand{\algorithmicrequire}{\textbf{Input:}}

\usepackage{amsmath,amssymb,amsfonts}

\usepackage{textcomp}
\usepackage{xcolor}

\ifCLASSINFOpdf
 \usepackage[pdftex]{graphicx}
 \usepackage{epstopdf}

\else
   
  \usepackage[dvips]{graphicx}
 
\fi
\ifCLASSOPTIONcompsoc
\usepackage[caption=false,font=normalsize,labelfont=sf,textfont=sf]{subfig}
\else
\usepackage[caption=false,font=footnotesize]{subfig}
\fi

\usepackage{graphicx}
\usepackage{color}
\usepackage{xfrac}
\usepackage{amsthm}
\usepackage{mathtools}
\usepackage{relsize}
\usepackage{enumerate}
\usepackage{booktabs}
\usepackage{morefloats}
\usepackage{cite}
\usepackage{bm}
\usepackage{stfloats}
\usepackage{setspace}
\usepackage{url}
\usepackage{ulem}
\usepackage{bbm}
\usepackage{setspace}

\usepackage{comment}
\usepackage{balance}
\usepackage{calc}
\usepackage{multirow}

\theoremstyle{plain}
\newtheorem{theorem}{Theorem}
\newtheorem{lemma}[]{Lemma}

\theoremstyle{remark}
\newtheorem{remark}[]{Remark}

\renewcommand{\algorithmicrequire}{\textbf{Input:}}

\begin{document}
\title{Achieving Stability for Aloha Networks with Multiple Receivers}

\author{Yunshan Yang and Lin Dai 
\thanks{This paper will be partly presented at the IEEE International Conference on Communications, Denver, USA, Jun. 2024.} 
\thanks{The authors are with the Department of Electrical Engineering, City University of Hong Kong, Hong Kong (e-mail: yunshyang3-c@my.cityu.deu.hk; lindai@cityu.edu.hk).}}

\IEEEaftertitletext{\vspace{-2.2\baselineskip}}




\maketitle


\IEEEpeerreviewmaketitle


\begin{abstract}
    Slotted Aloha has been widely adopted in various communication networks. Yet if the transmission probabilities and traffic input rates of transmitters are not properly regulated, their data queues may easily become unstable. For stability analysis of Aloha networks with multiple receivers, the focus of previous studies has been placed on the maximum input rate of each transmitter, below which the network is guaranteed to be stabilized under any given topology. By assuming a fixed and identical transmission probability across the network, however, network stability is found to be unachievable when the input rate exceeds zero.

    As we will demonstrate in this paper, the key to stabilizing the network lies in proper selection of transmission probabilities according to the traffic input rates and locations of all transmitters and receivers. Specifically, for an Aloha network with multiple capture receivers, by establishing and solving the fixed-point equations of the steady-state probabilities of successful transmissions of Head-of-Line (HOL) packets, the exact service rates of all transmitters' queues are obtained, based on which the operating region of transmission probabilities for achieving stability and the stability region of input rates are further characterized. The results are illustrated in various scenarios of multi-cell and ad-hoc networks. Simulation results validate the analysis and corroborate that the network can be stabilized as long as the traffic input rates are within the stability region, and the transmission probabilities are properly adjusted according to the traffic input rates and network topology.

\end{abstract}
\vspace{-0.1cm}

\vspace{-0.4cm}
\section{Introduction}\label{section: introduction}

Random access, which enables multiple transmitters to share a common channel with minimum coordination, provides a simple solution to embrace the burgeoning Machine-to-Machine (M2M) communications \cite{RandomaccessMTD_Leyva,ThechallengesofM2M_Biral}. Over the years, plenty of random access protocols have been developed \cite{DataNetworks_Gallager,WirelessCommunications_Goldsmith,ComputerNetworking_Kurose,Aloha_Abramson}. Among them, Aloha \cite{Aloha_Abramson}, with which each node transmits with a certain probability when it has packets in the buffer, has been widely adopted in various communication networks, such as 5G networks \cite{3GPPTS38.321}, Wi-Fi 6 networks \cite{802.11ax}, Long Range Radio Wide Area Networks (LoRaWAN) \cite{LoRa}, and Short Range Devices (SRD) systems \cite{SRD}.   


Wide as its applications are, due to the lack of coordination, the performance of an Aloha network may suffer from significant degradation if the transmissions are not well controlled. As concurrent transmissions may end up with failures, if high transmission probabilities are adopted, there could be severe contention and thus small chances of success. With low transmission probabilities, on the other hand, there could also be few successful packets due to few transmission attempts. Both may cause the service rate to drop below the input rate for each transmitter's data queue, in which case the queue length would grow unboundedly with time, rendering instability of the network. The difficulty of stabilizing a network may also grow as the traffic load becomes heavier. With the input rate reaching a critical value, a transmitter's data queue may not be stabilized no matter what transmission probability is chosen. 

Intuitively, for an Aloha network, there exists a stability region of input rates, only within which stability is achievable through transmission control. For given input rates, to stabilize the network, the transmission probabilities should further be chosen from a certain operating region. To characterize the stability region of input rates and the operating region of transmission probabilities for achieving stability, an analytical framework was recently proposed for an uplink single-cell Aloha network \cite{Atheoreticalframework_Dai}, and extended to a multi-cell Aloha network in \cite{Multicell_Yang} by further considering the inter-cell interference. 
 
In both \cite{Atheoreticalframework_Dai} and \cite{Multicell_Yang}, the collision receiver model \cite{Aloha_Abramson} was assumed, with which a packet transmission is successful only if there are no concurrent transmissions. Though capturing the essence of contention, the collision receiver could be overly pessimistic as concurrent transmissions may not necessarily fail especially when there exists a large difference in the received signal power. It is thus of great importance to extend the analysis to incorporate more advanced receiver models, such as the widely adopted capture model \cite{AlohaCapture_Roberts} where a packet can be successfully decoded as long as its received signal-to-interference-plus-noise ratio (SINR) is above a certain threshold.
Moreover, despite the multi-cell setting in \cite{Multicell_Yang}, an accurate characterization of the stability region of input rates and the operating region of transmission probabilities is available only for the two-cell case. For the general case with any number of receivers, how to obtain the stability region of input rates and control the transmission probabilities for achieving stability needs to be further studied.

In this paper, we will extend the stability analysis to incorporate multiple receivers and the capture receiver model. Note that in the literature, there have been considerable interests in modeling and analysis of Aloha networks with multiple capture receivers, including both the multi-cell and ad-hoc scenarios. For better understanding of the challenges in stability analysis for multi-receiver Aloha networks, let us first present a review on the related studies.

\vspace{-0.3cm}
\subsection{Modeling of an Aloha Network with Multiple Receivers}\label{subsection: relatedwork_modeling}

For modeling of an Aloha network with multiple receivers, existing studies mainly focused on modeling the random locations of transmitters and receivers in a large-scale network. For instance, an ad-hoc Aloha network was modeled as a Poisson bipolar network in \cite{AnAlohaprotocol_Baccelli}, where transmitters form a Poisson point process (PPP) with each of them having a dedicated receiver at a fixed distance. For a multi-cell Aloha network, on the other hand, both BSs and devices can be modeled as independent PPPs \cite{Onstochasticgeometry_ElSawy} or more complex point processes, e.g., a Poisson cluster process \cite{Effectspatialtemporaltrafficstatistics_Wang}.  

To analyze the performance of multiple transmitter-receiver (T-R) pairs, one key step is the characterization of the probability of successful transmissions of each transmitter, which crucially depends on the locations of all transmitters and receivers, and varies from T-R pair to T-R pair. Under the symmetric setting of identical transmission probabilities and input rates for all transmitters,\footnote[1]{Note that the assumptions on transmission power vary under different scenarios. For example, equal mean received SNR of all T-R pairs is assumed in \cite{AnAlohaprotocol_Baccelli,Statisticalprioritybased_Zhang,Stochasticopportunisticaloha_Baccelli,Onstochasticgeometry_ElSawy,SpatiotemporalStochastic_Gharbieh,Randomaccessanalysis_Jiang,RandomaccessPoisson_Stamatiou,SpatiotemporalMSA_Zhong} by properly adjusting the transmission power, while identical transmission power is assumed in \cite{grantfreeMIMO_Xia,Effectspatialtemporaltrafficstatistics_Wang}.}
the probability of successful transmissions of each T-R pair was usually approximated by the spatial average of the probability of successful transmissions of all T-R pairs \cite{AnAlohaprotocol_Baccelli,Onstochasticgeometry_ElSawy,Statisticalprioritybased_Zhang,Stochasticopportunisticaloha_Baccelli,grantfreeMIMO_Xia,Effectspatialtemporaltrafficstatistics_Wang,SpatiotemporalStochastic_Gharbieh,Randomaccessanalysis_Jiang,RandomaccessPoisson_Stamatiou,SpatiotemporalMSA_Zhong}, which can be obtained by leveraging tools from stochastic geometry \cite{Stochasticgeometryandrandomgraphs_Haenggi,StochasticGeometry_Baccelli,StochasticGeometry_Haenggi}. Despite a good approximation in the high-mobility scenario, e.g., locations of transmitters and receivers vary from time slot to time slot \cite{SpatiotemporalMSA_Zhong,RandomaccessPoisson_Stamatiou}, for the low-mobility or static scenario, which is common for M2M communications as the locations of machine-type devices seldom change after they are deployed, the spatial average probability of successful transmissions may significantly differ from the probability of successful transmissions of a given T-R pair. In those scenarios, the probability of successful transmissions of each T-R pair is highly sensitive to their locations as the interference level varies from T-R pair to T-R pair even under the symmetric setting of system parameters.

In \cite{Themetadistribution_Haenggi,Simpleapproximationsgeneralcellular_Kalamkar,SpatioTemporalCellCenterCellEdge_Yang,Perlinkreliability_Kalamkar,metacoverageprobabilityuplink_ElSawy,TheMetaCellularPowerControl_Wang,SINRmetaPoissoncellular_Feng}, by assuming identical transmission probabilities and input rates for all T-R pairs,
the spatial distribution of the probability of successful transmissions of each T-R pair was derived, based on which the proportion of T-R pairs that achieve certain target performance for a given SINR threshold can further be evaluated. Though helpful for the system-level performance analysis, the spatial distribution of the probability of successful transmissions cannot lead to accurate performance characterization of each specific T-R pair. Moreover, by assuming identical settings for all T-R pairs, the inherent asymmetry in practice cannot be captured, and more importantly, an individual adjustment of transmission probability of each T-R pair is ignored. As we will demonstrate in this paper, such individual adaptive transmission control based on the locations and traffic information of T-R pairs is indispensable for stabilizing the entire network.

\vspace{-0.5cm}
\subsection{Stability Analysis of an Aloha Network with Multiple Receivers}\label{subsection: introduction stability}
\vspace{-0.1cm}

For stability analysis of an Aloha network with multiple receivers, existing studies mainly focused on the ad-hoc scenario with a Poisson bipolar model \cite{AUnifiedFrameworkforSINR_Yang,OntheStabilityofStatic_Zhong,UncoordinatedMassive_Chisci,RandomaccessPoisson_Stamatiou}, or the downlink of a cellular network with independent PPP models for BSs and devices \cite{Spatiotemporalsmallcell_Yang}. Various sufficient or necessary conditions of the input rate for achieving network stability under any given topology have been developed with the symmetric setting of identical input rates and transmission probabilities of all T-R pairs across the network. 

Specifically, in \cite{RandomaccessPoisson_Stamatiou}, an upper-bound of the input rate of each T-R pair was derived, under which the network is guaranteed to be stabilized. The analysis was based on a key assumption that the probabilities of successful transmissions of all T-R pairs are equal to their spatial average, which holds only for the high-mobility scenario. 
For static networks where the probabilities of successful transmissions of T-R pairs may significantly vary with their locations even under the symmetric setting, the spatial distribution of probability of successful transmissions of each T-R pair was derived in \cite{AUnifiedFrameworkforSINR_Yang,OntheStabilityofStatic_Zhong,UncoordinatedMassive_Chisci,Spatiotemporalsmallcell_Yang}. Yet with a fixed and identical transmission probability across the network, to ensure that the network is stabilized under any topology, the input rate of each transmitter cannot exceed zero. That is why instead of stabilizing the whole network, the $\varepsilon$-stability was considered in \cite{AUnifiedFrameworkforSINR_Yang,OntheStabilityofStatic_Zhong,UncoordinatedMassive_Chisci,Spatiotemporalsmallcell_Yang}, where the target is that the percentage of stable transmitters is no smaller than $1-\varepsilon$ for a small $\varepsilon > 0$.

As we will demonstrate in this paper, for given network topology, \textit{all} the transmitters' queues can indeed be stabilized as long as their transmission probabilities are properly selected according to their input rates and locations. The key lies in the characterization of the probability of successful transmissions of each T-R pair, which is crucially dependent on the locations of its transmitter and receiver as well as its traffic input rate and transmission probability. Based on the exact probabilities of successful transmissions of all T-R pairs, the operating region of transmission probabilities for achieving stability and stability region of input rates can further be obtained.

\vspace{-0.2cm} 
\subsection{Our Contributions}\label{subsection:contributions}
\vspace{-0.05cm}

In this paper, the stability analysis is extended from \cite{Atheoreticalframework_Dai} to incorporate multiple capture receivers for characterizing the stability region of input rates and operating region of transmission probabilities for achieving stability in Aloha networks. Specifically, based on the proposed Head-of-Line (HOL)-packet model of each transmitter, fixed-point equations of the probabilities of successful transmissions of all transmitters $\bm{p}$ are established and solved to obtain the network steady-state point for given traffic input rates, transmission probabilities, and locations of all T-R pairs. 

For a stable network, all the transmitters' queues should be unsaturated with the service rate of each queue larger than its input rate. To ensure that the network operates at the all-unsaturated steady-state point $\bm{p}_L$ for given input rates, the transmission probabilities $\bm{q}$ should be selected from a certain region, i.e., the all-unsaturated region, to characterize which a multi-objective optimization problem is formulated to find the maximum and minimum of $\bm{q}$ with the constraint of $\bm{p} = \bm{p}_L$. 
For the stability region of input rates, within which the network stability can always be achieved by properly adjusting transmission probabilities, it is indeed the maximum input rates under the constraint that the all-unsaturated region of transmission probabilities is non-empty, and can be obtained based on the Pareto-optimal solutions of the constrained multi-objective optimization problem. 

Although in general, the network steady-state point, the all-unsaturated region of transmission probabilities, and the stability region of input rates can only be calculated numerically, it is shown that explicit expressions are available in two special cases, i.e., the two T-R pairs and $K$ symmetric T-R pairs, from which the effects of key system parameters can be clearly observed. Simulation results verify the analysis in various network topologies, and corroborate the importance of performing individual adaptive transmission control for achieving network stability.

The remainder of this paper is organized as follows. The system model and key assumptions are presented in Section \ref{section: model}. In Section \ref{section: steady-state point}, the network steady-state point is obtained based on the HOL-packet model and fixed-point equations of the probabilities of successful transmissions. The all-unsaturated region of transmission probabilities and the stability region of input rates are further characterized in Section \ref{section:transmission control} and Section \ref{section: stability region}, respectively. In Section \ref{section: simulation}, simulation results are presented to verify the analysis. Finally, concluding remarks are given in Section \ref{section: conclusion}.

Throughout the paper, for vectors $\bm{x} = (x_1,\ldots,\:x_n)$ and $\bm{y} = (y_1,\ldots,\: y_n)$, $\bm{x} < \bm{y}$ denotes the element-wise inequality, i.e., $x_i < y_i$ for $i = 1,\ldots,\:n$. $|\mathcal{X} \vert$ denotes the cardinality of set $\mathcal{X}$.

\vspace{-0.4cm}
\section{System Model}\label{section: model}
Consider a slotted Aloha network, which contains $K$ transmitters and $L$ receivers. Let $\mathcal{K} = \{1,\ldots,K\}$ and $\mathcal{L} = \{1,\ldots,L\}$ denote the set of transmitters and the set of receivers, respectively. Note that $K$ and $L$ vary under different network scenarios. For the uplink of a single-cell network where multiple transmitters communicate with a single receiver as Fig. \ref{subfig:multiaccess_compare} illustrates, $K>L = 1$. For the uplink of a multi-cell network, $K>L >1$, as Fig. \ref{subfig:multicell_compare} shows. For the downlink of a multi-cell network, or an ad-hoc network, where each transmitter has one dedicated receiver, as Figs. \ref{subfig:DL_multicell_compare} and \ref{subfig:AdHoc_compare} show, respectively, $K = L > 1$.

Assume that each transmitter has a data buffer to store incoming packets, and transmits its HOL packet to its receiver with a certain probability $q_k$ in each time slot if its data buffer is not empty, $k\in\mathcal{K}$. Assume that the transmission probability vector $\bm{q}=(q_1,\ldots, q_K)$ does not change with time, and a HOL packet stays in the queue until it is successfully transmitted.    

For the channel model, let $g_{k,l} = \gamma_{k,l}\cdot h_{k,l}$ denote the channel gain between Transmitter $k \in\mathcal{K}$ and Receiver $l\in\mathcal{L}$. $\gamma_{k,l} = d_{k,l}^{-\alpha/2}$ is the large-scale fading coefficient between Transmitter $k$ and Receiver $l$, where $d_{k,l} $ is the distance between Transmitter $k$ and Receiver $l$, and $\alpha$ is the path-loss exponent. $h_{k,l}$ is the small-scale fading coefficient that varies from time slot to time slot and is modeled as a complex Gaussian random variable with zero mean and unit variance. 
The mean received signal-to-noise ratio (SNR) of Transmitter $k$ at Receiver $l$, denoted as $\rho_{k,l}$, is then given by $\rho_{k,l} = P_k|\gamma_{k,l}|^2/\sigma^2$, where $P_{k}$ denotes the transmission power of $k$, and $\sigma^2$ denotes the variance of the additive white Gaussian noise (AWGN).  

At the receivers, the capture model is assumed. That is, each transmitter's packet is decoded independently by treating others' transmissions as background noise at each time slot. A packet of Transmitter $i$ can be successfully decoded as long as its received signal-to-interference-plus-noise ratio (SINR) at its receiver, denoted as $i^{\ast}$, is equal to or above a certain threshold, denoted as $\theta_{i^{\ast}}$, i.e., 
\begin{equation}\label{SINR_exp}
    \text{SINR}_i = \tfrac{P_{i} |g_{i,i^{\ast}}|^2}{\sum_{j\in \mathcal{S}_i}P_j |g_{j,i^{\ast}}|^2 + \sigma^2} \geq \theta_{i^{\ast}} 
\end{equation}
where $\mathcal{S}_i$ is the set of transmitters that have concurrent transmissions with Transmitter $i$, $i \in\mathcal{K}$.

\begin{figure}[!t]
    \vspace{-0.3cm}
    \captionsetup[subfigure]{justification=centering}
    \centering
    \subfloat[]{
        \includegraphics[width=0.6in,height=0.55in]{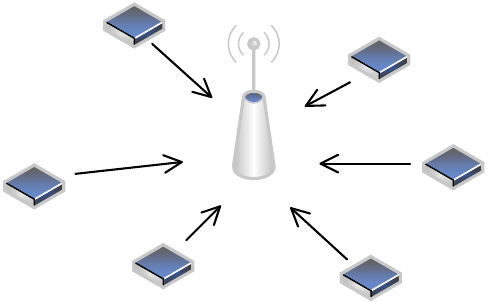}
    \label{subfig:multiaccess_compare}}
    \hfill
    \subfloat[]{
        \includegraphics[width=0.9in,height=0.55in]{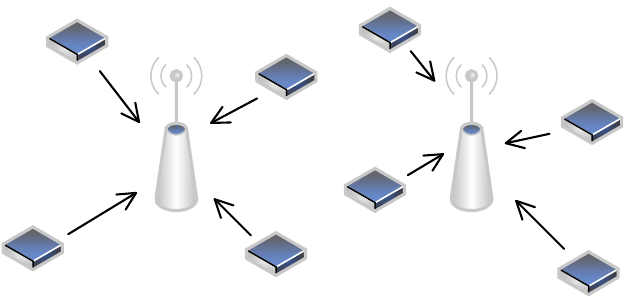}
    \label{subfig:multicell_compare}}
    \hfill
    \subfloat[]{
        \includegraphics[width=0.9in,height=0.55in]{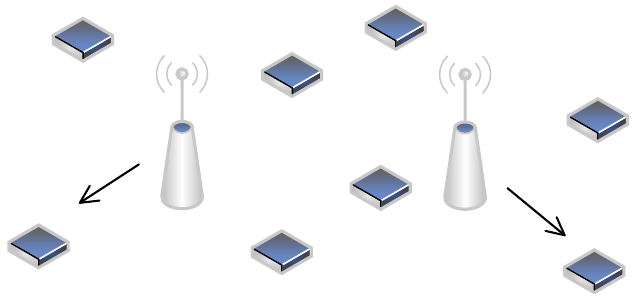}
    \label{subfig:DL_multicell_compare}}
    \hfill
    \subfloat[]{
        \includegraphics[width=0.6in,height=0.55in]{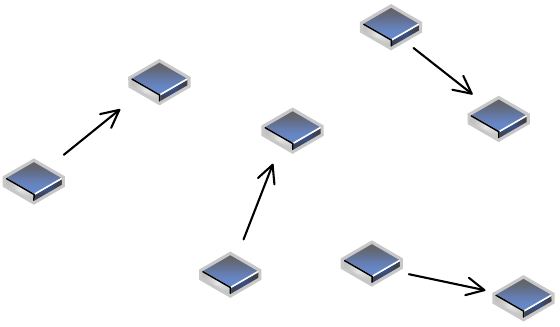}
    \label{subfig:AdHoc_compare}}
    \vspace{-0.1cm}
    \caption{Graphic illustration of (a) the uplink of a single-cell network, (b) the uplink of a multi-cell network,  (c) the downlink of a multi-cell network, and (d) an ad-hoc network. }
    \label{fig:multiaccess_AdHoc_compare}
    \vspace{-0.5cm}
\end{figure}

In this paper, we consider a static network with fixed locations of transmitters and receivers. Note that even with fixed locations, the $\text{SINR}$ of each packet is still random due to the random channel gains and uncoordinated transmissions. With each transmitter transmitting its data packets independently with a certain probability in each time slot, it can be expected that the network performance is crucially determined by the transmission probabilities $\bm{q}$, and may significantly degrade if $\bm{q}$ are not properly selected.

\vspace{-0.55cm}
\subsection{Network Stability}
\vspace{-0.15cm}

For a transmitter's queue, define the input rate and the service rate as the long-term average number of packets that arrive per time slot, and that are served per time slot when the queue is not empty, respectively. Its throughput is defined as the long-term average number of served (successfully transmitted) packets per time slot.
Similar to \cite{Atheoreticalframework_Dai,Multicell_Yang,OntheStabilityofStatic_Zhong,AUnifiedFrameworkforSINR_Yang,UncoordinatedMassive_Chisci,RandomaccessPoisson_Stamatiou,Spatiotemporalsmallcell_Yang}, define that a queue is stable if the steady-state distribution of its queue length $Q(t)$ exists, i.e., 
\begin{equation}
    \lim_{t\to \infty} \Pr\{Q(t) < x\} = F(x) \quad \text{and} \quad \lim_{x\to \infty} F(x) = 1, 
\end{equation}
which is satisfied when the input rate is smaller than the service rate for stationary arrival and service processes \cite{Thestabilityofaqueue_Loynes}. A network is stable if all the queues in the network are stable.   

It can be seen that the stability performance of the network is determined by the input rate and service rate of each transmitter's queue. When the input rate is below the service rate, the queue is unsaturated with a non-zero probability of being empty, and the throughput is equal to the input rate. On the other hand, with the input rate reaching or exceeding the service rate, the queue is busy with probability 1, in which case the queue is saturated and its throughput is equal to the service rate. A queue is stable only when it is unsaturated, and the network is stable only when all the queues are unsaturated.

For each Transmitter $k \in \mathcal{K}$, denote the input rate, service rate and throughput of its queue as $\lambda_k$, $\mu_k$, and $\lambda_{out,k}$, respectively. Further let $\bm{\lambda} = (\lambda_1,\ldots, \lambda_K)$, $\bm{\mu} = (\mu_1,\ldots, \mu_K)$, and $\bm{\lambda}_{out} = (\lambda_{out, 1},\ldots, \lambda_{out, K})$ denote the input rate vector, service rate vector, and throughput vector of the network, respectively. The network is stable only when all the queues are unsaturated with $\bm{\lambda} < \bm{\mu}$, in which case we have $\bm{\lambda}_{out} = \bm{\lambda}$.

\vspace{-0.3cm}
\subsection{Transmission Control}\label{subsection: transmission control_system model}
\vspace{-0.05cm}

With Aloha, the service rate of each transmitter's queue is determined by the transmission probability $q_k$, $k\in\mathcal{K}$. For given network topology  and the input rates $\bm{\lambda}$, we are interested in stabilizing the network by jointly controlling the transmission probabilities $\bm{q}$ of all transmitters. 
Specifically, define the all-unsaturated region of $\bm{q}$ as\footnotemark[2]\footnotetext[2]{Note that $S_{U^K}(\bm{q},\bm{\lambda})$ can also be regarded as a region of input rates $\bm{\lambda}$ for given transmission probabilities $\bm{q}$. That is why we include both $\bm{q}$ and $\bm{\lambda}$ in (\ref{define_UU}), although in this paper we mainly focus on the transmission control, i.e., adjusting transmission probabilities $\bm{q}$ based on the input rates $\bm{\lambda}$.}
\vspace{-0.1cm}
\begin{equation}\label{define_UU}
    S_{U^K}(\bm{q},\bm{\lambda}) \triangleq \{\bm{q} :  \bm{\lambda} < \bm{\mu}  \}.
    \vspace{-0.1cm}
\end{equation}
The network stability can be achieved if $S_{U^K}(\bm{q},\bm{\lambda}) \neq \varnothing$ and $\bm{q}$ are chosen from $S_{U^K}(\bm{q},\bm{\lambda})$. It can be seen that to obtain the all-unsaturated region $S_{U^K}(\bm{q},\bm{\lambda})$, the key lies in the characterization of the service rates $\bm{\mu}$ of transmitters' queues. 
In Section \ref{section:transmission control}, we will demonstrate how to find $S_{U^K}(\bm{q},\:\bm{\lambda})$ by deriving $\bm{\mu}$ as a function of the steady-state probabilities of successful transmission of HOL packets of transmitters' queues.

\vspace{-0.3cm}
\subsection{Stability Region of Input Rates}\label{subsection: stability region_system model}
\vspace{-0.05cm}

To ensure that the network can be stabilized by choosing transmission probabilities from the all-unsaturated region $S_{U^K}(\bm{q},\bm{\lambda})$, we need $S_{U^K}(\bm{q},\bm{\lambda}) \neq \varnothing$, which can be satisfied only when the input rates $\bm{\lambda}$ are small enough.  We are interested in characterizing the stability region of input rates $\bm{\lambda}$, inside which there exists some transmission probability vector to achieve network stability. Let $S_Q(\bm{\lambda})$ denote the stability region of $\bm{\lambda}$, which can be defined as 
\vspace{-0.1cm}
\begin{equation}\label{define_SQ}
    S_Q(\bm{\lambda} )\triangleq \bigcup_{\bm{q}}  \{ \bm{\lambda}: \bm{\lambda} < \bm{\mu} \} .
    \vspace{-0.2cm}
\end{equation}
With $\bm{\lambda} \in S_Q(\bm{\lambda})$, the all-unsaturated region $S_{U^K}(\bm{q},\bm{\lambda}) \neq \varnothing$ and the network can be stabilized by further choosing $\bm{q}$ from $S_{U^K}(\bm{q},\bm{\lambda})$.

\begin{figure}[!t]
    \centering
    \includegraphics[width=2.1in,height=1.3in]{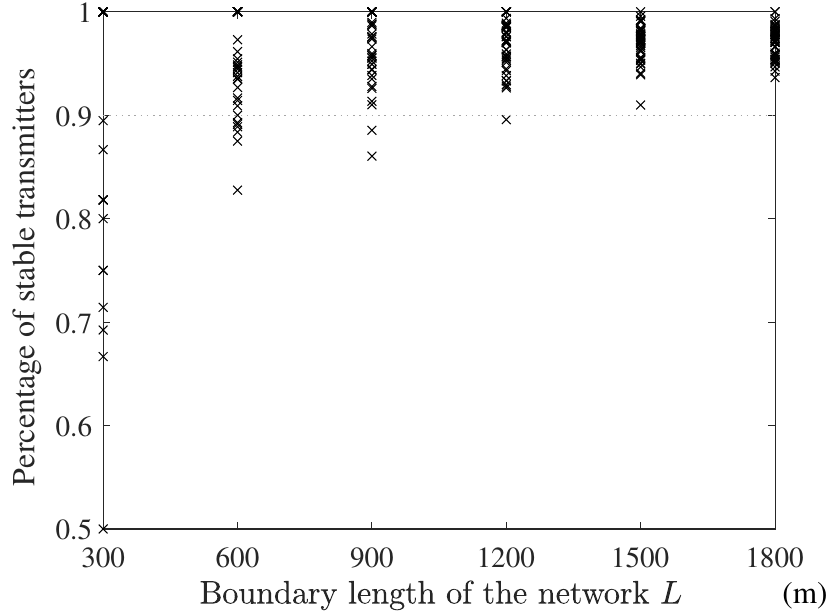}
    \vspace{-0.3cm}
    \caption{The percentage of stable transmitters versus the boundary length $L$ of the network. Transmitters' locations are generated in $[0, L]^2$ according to PPP with density $\xi=10^{-4}$ m$^{-2}$. Each transmitter has a receiver with $d_{i,i}= 25$ m with random orientation, $i\in \mathcal{K}$. $\lambda_i = 0.2$,  $q_i = 1$, $\theta_i = 0$ dB, $P_i = 17$ dBm, $i\in\mathcal{K}$, $\sigma^2 =-90$ dBm, and $\alpha = 3.8$.}
    \label{fig:ratio_stable_nodes_vs_boundlength_sample50_thres1_transP17dBm_irate0.2_q1_YangHoward}
    \vspace{-0.5cm}
\end{figure}

\begin{remark}
    \textit{Note that for previous studies on stability analysis of Aloha networks with multiple capture receivers, various sufficient conditions or necessary conditions of input rates $\bm{\lambda}$ were obtained for achieving $\varepsilon$-stability\footnotemark[3]\footnotetext[3]{Note that for $\varepsilon$-stability, a network is defined as stable if the percentage of stable transmitters is no smaller than $1-\varepsilon$.} given the transmission probabilities $\bm{q}$ \cite{OntheStabilityofStatic_Zhong,AUnifiedFrameworkforSINR_Yang,UncoordinatedMassive_Chisci,Spatiotemporalsmallcell_Yang}. For instance, in \cite{AUnifiedFrameworkforSINR_Yang}, based on a Poisson bipolar model for T-R pairs, it was shown that with density $\xi=10^{-4}$ m$^{-2}$, T-R distance $d_{i,i} = 25$ m, SINR threshold $\theta_i = 0$ dB, and transmission power $P_i = 17$ dBm for all $i\in\mathcal{K}$, $90\%$ of transmitters in the network can be stabilized if the input rate and transmission probability of each transmitter are set to $\lambda_i = 0.2$ and $q_i = 1$, respectively, $i\in\mathcal{K}$.}    

    \textit{To verify the result, we generate 50 topologies of T-R pairs within a square area of side length $L$ varying from 300 m to 1800 m and count the percentage of stable transmitters in each case. It can be observed from Fig. \ref{fig:ratio_stable_nodes_vs_boundlength_sample50_thres1_transP17dBm_irate0.2_q1_YangHoward} that only when the area is sufficiently large, e.g., $L\geq 1500$ m, can the percentage reach $90\%$. The reason is that the results were obtained asymptotically by assuming that the network area goes to infinity. With a small $L$, e.g., $L = 300$ m, the percentage of stable transmitters may be significantly lower than the target $90\%$ with the proposed setting. As we will demonstrate in this paper, $100\%$ of the transmitters' queues can indeed be stabilized as long as the transmission probabilities $\bm{q}$ are selected from the all-unsaturated region $S_{U^K}(\bm{q},\bm{\lambda})$ for given T-R locations, and the input rates $\bm{\lambda}$ are within the stability region $S_Q(\bm{\lambda} )$. }  
\end{remark}

\vspace{-0.1cm}
\section{HOL-Packet Modeling and Steady-State Probabilities of Successful Transmission of HOL Packets}\label{section: steady-state point}
\vspace{-0.05cm}

As we can see from (\ref{define_UU}), the key to stabilizing the network lies in the characterization of the service rate of each transmitter's queue, which depends on the activities of HOL packets of all the transmitters. In this section, we will first review the modeling of HOL packets in a single-cell Aloha network, and then extend the analysis to the scenario with multiple receivers.      

\vspace{-0.45cm}
\subsection{HOL-Packet Modeling} 
\vspace{-0.05cm}

In \cite{Atheoreticalframework_Dai}, a discrete-time Markov renewal process $(\mathbf{X}_k,\mathbf{V}_k) = \{(X_{k,j}, V_{k,j}),  j = 0,1,\ldots \}$ was built to model the behavior of HOL packets of Transmitter $k$ in a single-cell Aloha network, i.e., multiple transmitters transmit to a single receiver, where $X_{k,j}$ denotes the state of a tagged HOL packet of Transmitter $k$ at the $j$-th transition and $V_{k,j}$ denotes the epoch at which the $j$-th transition occurs, $k\in\mathcal{K}$. Fig. \ref{fig:StateTranDiagram} illustrates the embedded Markov chain $\mathbf{X}_k$. A fresh HOL packet stays in the initial State $T$ if it is successfully transmitted.  Otherwise, it moves to State $B$. $p_{k,t}$ denotes the probability of successful transmission of HOL packets of Transmitter $k$ at time slot $t$. Let $p_k = \lim_{t\to \infty} p_{k,t}$. The steady-state probability distribution of the embedded Markov chain shown in Fig. \ref{fig:StateTranDiagram} can be obtained as $\pi_T^k = p_k q_k$ and $\pi_B^k = 1- p_k q_k$. Let $\tau_m^k $ denote the mean holding time at State $m \in \{T,B\}$. As each packet transmission, successful or failed, lasts for one time slot, we have $\tau_T^k = \tau_B^k = 1$ time slot. The limiting state probabilities of the Markov renewal process can be obtained as $\tilde{\pi}_m^k = \pi_m^k$, $m\in \{T,B\}$.

\begin{figure}[!t]
    \centering
    \includegraphics[width=0.38\textwidth,height =0.05\textwidth]{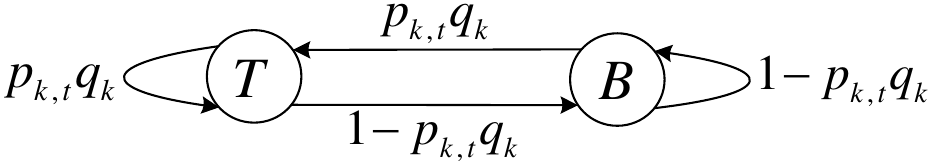}
    \vspace{-0.3cm}
    \caption{Embedded Markov chain $\mathbf{X}_k$ of the state transition process of HOL packets of Transmitter $k$, $k\in\mathcal{K}$.}  \label{fig:StateTranDiagram}
    \vspace{-0.55cm}
  \end{figure}

  Note that $\tilde{\pi}_T^k$ is the service rate of Transmitter $k$ as its queue has a successful output if and only if the HOL packet stays at State $T$. We have
  \vspace{-0.35cm}
  \begin{equation}\label{service_rate}
    \mu_k = \tilde{\pi}_T^k = q_k p_k,
    \vspace{-0.1cm}
  \end{equation} 
$k\in\mathcal{K}$. It can be seen from (\ref{service_rate}) that the service rate $\mu_k$ is crucially determined by $p_k$, the steady-state probability of successful transmission of HOL packets of Transmitter $k$. 

\vspace{-0.3cm}
\subsection{Steady-State Probabilities of Successful Transmission of HOL Packets}\label{subsection: steady-state probabilities}

Let $\bm{p} = (p_1,\ldots,p_K)$ denote the steady-state probabilities of successful transmission of HOL packets of $K$ transmitters. To obtain $p_i$, let us first consider the conditional steady-state probability of HOL packets of Transmitter $i$ for given set of transmitters $\mathcal{S}_i$ that have concurrent transmissions, denoted as $r_{i|\mathcal{S}_i}$. According to (\ref{SINR_exp}), we have 
\vspace{-0.15cm}
\begin{align}\label{condition_p_final}
    r_{i|\mathcal{S}_i} & = \Pr\left\{\tfrac{ |h_{i,i^{\ast}}|^2  }{ \sum_{j\in\mathcal{S}_i} |h_{j,i^{\ast}}|^2 \cdot \tfrac{\rho_{j,i^{\ast}}}{\rho_{i,i^{\ast}}} + \tfrac{1}{\rho_{i,i^{\ast}}}} \geq \theta_{i}^{\ast} \right\} \nonumber\\
    & = \tfrac{\exp\left( -\tfrac{\theta_{i^{\ast}}}{\rho_{i,i^{\ast}}}\right)}{\prod_{j\in \mathcal{S}_i} \left(1 + \tfrac{\rho_{j,i^{\ast}}}{\rho_{i,i^{\ast}}} \cdot \theta_{i^{\ast}}\right) },
    \vspace{-0.1cm}
\end{align}
where $i^{\ast} \in \mathcal{L}$ is the receiver that Transmitter $i$ is associated with, $i\in\mathcal{K}$.  The steady-state probability of successful transmission of HOL packets of Transmitter $i$ can then be written as
\vspace{-0.2cm}
\begin{equation}\label{p_middle}
    p_i = \sum_{\mathcal{S}_i \subseteq  \mathcal{K} \setminus \{i\}} r_{i|\mathcal{S}_i} \cdot \prod_{j \in \mathcal{S}_i} x_j \cdot \prod_{k \in \mathcal{K} 
    \setminus \left(\mathcal{S}_i \cup \{i\}\right)} (1 - x_k),
    \vspace{-0.1cm}
\end{equation}
$i \in\mathcal{K}$, where $x_j$ denotes the probability that Transmitter $j$ is requesting a transmission. If Transmitter $j$ is saturated, i.e., $\mu_j \leq \lambda_j$, then its queue is busy with probability 1, in which case $x_j = q_j$. If Transmitter $j$ is unsaturated, i.e., $\mu_j > \lambda_j$, its queue is busy with probability $\tfrac{\lambda_j}{\mu_j}$, in which case $x_j = \tfrac{\lambda_j}{\mu_j}\cdot q_j = \tfrac{\lambda_j}{p_j}$. We then have
\vspace{-0.1cm}
\begin{equation}\label{x_requesting}
    x_j = \min \left( \tfrac{\lambda_j}{p_j}, q_j \right),
\end{equation}
for $j \in \mathcal{K}$. By combining (\ref{condition_p_final}) and (\ref{x_requesting}), (\ref{p_middle}) can be further written as 
\vspace{-0.25cm}
\begin{equation}\label{p}
    p_i = \exp\left(-\tfrac{\theta_{i^{\ast}}}{\rho_{i,i^{\ast}}} \right) \cdot \prod_{j\in \mathcal{K}\setminus\{i\}} \left(1 - \tfrac{\theta_{i^{\ast}}}{\theta_{i^{\ast}} + \tfrac{\rho_{i,i^{\ast}}}{\rho_{j,i^{\ast}}}} \cdot \min \left( \tfrac{\lambda_j}{p_j}, q_j \right) \right), 
    \vspace{-0.2cm}
\end{equation}
for $i \in \mathcal{K}$. 

(\ref{p}) indicates that the fixed-point equation of $p_i$ depends on whether each transmitter is saturated or not. Let $\phi_i\in\{U, S\}$ denote whether Transmitter $i$'s queue is unsaturated or saturated, $i\in\mathcal{K}$. The state of the network can then be represented by $\bm{\phi} = (\phi_1,\ldots, \phi_K)$, and there are $2^K$ possible network states in total. Given that the network operates at State $\bm{\phi}$, the fixed-point equation of $\bm{p}$ given in (\ref{p}) can be written as $\bm{p} = \bm{f}_{\bm{\phi}}(\bm{p})$, where $\bm{f}_{\bm{\phi}}(\bm{p}) = (f_{1,\bm{\phi}}(\bm{p}), \ldots,  f_{K,\bm{\phi}}(\bm{p}) )$ and 
\vspace{-0.2cm}
\begin{multline}\label{p_given_state}
    f_{i,\bm{\phi}}(\bm{p}) = \exp\left(-\tfrac{\theta_{i^{\ast}}}{\rho_{i,i^{\ast}}} \right) \cdot \prod_{j\in \mathcal{K}\setminus\{i\}, \phi_j = S} \left(1 - \tfrac{\theta_{i^{\ast}}}{\theta_{i^{\ast}} + \tfrac{\rho_{i,i^{\ast}}}{\rho_{j,i^{\ast}}}} \cdot q_j  \right) \\
    \cdot \prod_{j\in \mathcal{K}\setminus\{i\}, \phi_j = U} \left(1 - \tfrac{\theta_{i^{\ast}}}{\theta_{i^{\ast}} + \tfrac{\rho_{i,i^{\ast}}}{\rho_{j,i^{\ast}}}} \cdot \tfrac{\lambda_j}{p_j}\right),
    \vspace{-0.2cm}
\end{multline}
for $i\in\mathcal{K}$. The corresponding network steady-state point at the given State $\bm{\phi}$ can then be obtained as the attracting fixed point of $\bm{p} = \bm{f}_{\bm{\phi}}(\bm{p})$. For instance, if the network is all-saturated, i.e., $\phi_i = S$ for all $i\in\mathcal{K}$, the network steady-state point, denoted as $\bm{p}_A= (p_{1,A},\ldots, p_{K, A})$, can be obtained from (\ref{p_given_state}) as 
\begin{equation}\label{p_S^K}
    p_{i,A} = \exp\left(-\tfrac{\theta_{i^{\ast}}}{\rho_{i,i^{\ast}}} \right) \cdot \prod_{j\in \mathcal{K}\setminus\{i\}
    } \left(1 - \tfrac{\theta_{i^{\ast}}}{\theta_{i^{\ast}} + \tfrac{\rho_{i,i^{\ast}}}{\rho_{j,i^{\ast}}}} \cdot q_j  \right), 
    \vspace{-0.1cm}
\end{equation}
$i \in \mathcal{K}$, in which case the network steady-state point $\bm{p}_A$ is determined by the transmission probabilities $\bm{q}$ of all transmitters.

In general, if only the input rates $\bm{\lambda}$ and transmission probabilities $\bm{q}$ are given without knowing the state of the network, to determine the network steady-state point $\bm{p}$, we may obtain the attracting fixed point of $\bm{p} = \bm{f}_{\bm{\phi}}(\bm{p})$ for all possible network states until (\ref{p}) holds. Specifically, for network state $\bm{\phi}$, the fixed-point equation $\bm{p} = \bm{f}_{\bm{\phi}}(\bm{p})$ can be solved iteratively with the initial value $\bm{p}(1)$ randomly chosen, and only the attracting fixed point\footnotemark[4]\footnotetext[4]{Note that for the fixed-point equation $\bm{p} = \bm{f}_{\bm{\phi}}(\bm{p})$, multiple fixed points may exist. To determine whether a fixed point of $\bm{p} = \bm{f}_{\bm{\phi}}(\bm{p})$, denoted as $\bm{\tilde{p}}$, is attracting or not, the spectral radius of the Jacobian matrix of $\bm{f}_{\bm{\phi}}(\bm{\tilde{p}})$ can be calculated. $\bm{\tilde{p}}$ is an attracting fixed point if and only if the spectral radius is smaller than 1 \cite{GlobalBehavior_Kocic}.} is the output. Starting from the all-saturated state where $\phi_i = S$ for all $i\in \mathcal{K}$, the network steady-state point for each possible network state is calculated, and the search process is terminated when (\ref{p}) holds. The detailed steps are presented in Algorithm \ref{alg:iter}.

\begin{algorithm}[t!]
    \caption{Calculation of the Steady-state Probabilities of Successful Transmission of HOL Packets $\bm{p}$}\label{alg:iter}
    \begin{algorithmic}[1]
      \Require Input rate $\lambda_i$, transmission probability $q_i$ for all $i\in\mathcal{K}$, SINR threshold $\theta_i$ for all $i\in\mathcal{L}$, and mean received SNR $\rho_{i,j}$ for all $i\in\mathcal{K}$, $j\in\mathcal{L}$.

      \State \textit{findroot} $ \gets $ \textit{false}, $k \gets 1$, $\bm{p} \gets \bm{p}_A$
      \If {(\ref{p}) holds}    
          \State \textit{findroot} = \textit{true}      
      \EndIf
      \While {\textit{findroot} = \textit{false}}  
          \State $\mathbf{A}_k \gets $ $ \{\bm{\phi} : \sum_{i =1}^K \mathbbm{1}_{\phi_i = U} = k\text{, where }  \mathbbm{1}_{\phi_i = U} = 1\text{ if } $ 
          
          $\!\!\!\phi_i = U\text{, and } \mathbbm{1}_{\phi_i = U} = 0 \text{ if } \phi_i = S\} $
          \ForAll {$\bm{\phi}\in \mathbf{A}_k$}
                \State  $r \gets 1$  
                    \While {$r \geq 1$}
                        \State $t \gets 1$ 
                        \State For all $i\in\mathcal{K}$, set $p_i(t)$ to a random number 
                        
                        $\quad \quad\:$ between $0$ and $1$
                        \State $\bm{p}(t + 1) \gets \bm{f}_{\bm{\phi}}(\bm{p}( t)) $
                        \While{$\bm{p}(t+1) \neq \bm{p}(t)$}  
                            \State $t \gets t +1$, $\bm{p}(t + 1) \gets \bm{f}_{\bm{\phi}}(\bm{p}( t)) $
                        \EndWhile
                        \State $r \gets $ the spectral radius of the Jacobian matrix 
                        
                        $\quad \quad\:$ of $\bm{f}_{\bm{\phi}} (\bm{p}(t))$
                        \If {$r < 1$ \textbf{and} (\ref{p}) holds}
                            \State \textit{findroot} $\gets$ \textit{true}, $\bm{p} \gets \bm{p}(t)$
                        \EndIf                                   
                    \EndWhile
                    \If {\textit{findroot} = \textit{true}}
                        \State \textbf{break}
                    \EndIf
          \EndFor
          \State $k \gets k + 1$
      \EndWhile
      \Ensure Steady-state probabilities of successful transmission of HOL packets $\bm{p}$.
    \vspace{-0.01cm} 
    \end{algorithmic}
\end{algorithm}

Note that Algorithm \ref{alg:iter} provides a numerical method for calculating the network steady-state point in general. In some special cases, however, explicit expressions of the network steady-state points can be obtained. In the following, we will take the example of two T-R pairs and $K$ symmetric T-R pairs to demonstrate how to obtain the network steady-state points.
  
\vspace{-0.3cm}
\subsection{Special Cases} \label{subsection:special cases_p}

\subsubsection{Two T-R Pairs}\label{subsubsection:2_TR_p}
With $K = L = 2$,  the network has $4$ possible states: all-saturated with State $SS$, partially-saturated with State $SU$ or $US$, and all-unsaturated with State $UU$.

When both transmitters are saturated, i.e., $\phi_i = S$ for $i\in\{1,2\}$, according to (\ref{p}), the all-saturated steady-state point $\bm{p}^{K=2}_{A}$ is obtained as 
\vspace{-0.2cm}
\begin{multline}\label{p_twoTR_SS}
    \bm{p}^{K=2}_{A} = \Biggl( \exp\left(-\tfrac{\theta_{1}}{\rho_{1,1}} \right) \left(1 - \tfrac{\theta_{1}}{\theta_{1} + \tfrac{\rho_{1,1}}{\rho_{2,1}}} \cdot  q_2 \right), \\
    \exp\left(-\tfrac{\theta_{2}}{\rho_{2,2}} \right)  \left(1 - \tfrac{\theta_{2}}{\theta_{2} + \tfrac{\rho_{2,2}}{\rho_{1,2}}} \cdot  q_1 \right) \Biggr).
    \vspace{-0.2cm}
\end{multline}

When Transmitter $i$ is unsaturated and Transmitter $j$ is saturated, i.e., $\phi_i = U$ and $\phi_j = S$, $i, j\in \{1, 2\}$, $i\neq j$, the partially-saturated steady-state point $\bm{p}^{K=2}_{P}$ can be obtained from (\ref{p}) as
\vspace{-0.2cm}
\begin{equation}\label{pi_twoTR_US}
    p^{K=2}_{i,P} = \exp\left(-\tfrac{\theta_{i}}{\rho_{i,i}} \right)  \left(1 - \tfrac{\theta_{i}}{\theta_{i} + \tfrac{\rho_{i,i}}{\rho_{j,i}} } \cdot  q_j \right),
    \vspace{-0.2cm}
\end{equation}
and 
\vspace{-0.1cm}
\begin{multline}\label{pj_twoTR_US}
    p^{K=2}_{j,P}
     = \exp\left(-\tfrac{\theta_{j}}{\rho_{j,j}} \right) \\
     \cdot\left(1 - \tfrac{\theta_{j}}{\theta_{j} + \tfrac{\rho_{j,j}}{\rho_{i,j}}} \cdot \tfrac{\lambda_i}{1 - \tfrac{\theta_{i}}{\theta_{i} + \rho_{i,i}/\rho_{j,i}} \cdot  q_j} \cdot \exp\left(\tfrac{\theta_{i}}{\rho_{i,i}} \right) \right).
     \vspace{-0.1cm}
\end{multline} 
 
When both transmitters are unsaturated, i.e., $\phi_i = U$ for $i\in\{1,2\}$,  (\ref{p}) can be written as
\begin{equation}\label{p_2TR_UU}
    p_i = \exp\left(-\tfrac{\theta_{i}}{\rho_{i,i}} \right)   \left(1 - \tfrac{\theta_{i}}{\theta_{i} + \tfrac{\rho_{i,i}}{\rho_{j,i}}} \cdot  \tfrac{\lambda_j}{p_j} \right), 
    \vspace{-0.1cm}
\end{equation}
for $i, j\in \{1, 2\}$, $i\neq j$. Let 
\vspace{-0.1cm}
\begin{equation}\label{a_b}
    a_i = \exp\left(-\tfrac{\theta_{i}}{\rho_{i,i}}\right), \quad b_i = \tfrac{\theta_{j}}{\theta_{j} + \tfrac{\rho_{j,j}}{\rho_{i,j}}}, 
    \vspace{-0.1cm}
\end{equation}
$i, j\in \{1, 2\}$, $i\neq j$, and 
\vspace{-0.1cm}
\begin{equation}\label{c_K=2}
    c = \prod_{i=1}^{2} \left( 1 - \tfrac{b_i \lambda_i}{p_i}\right).
    \vspace{-0.1cm}
\end{equation}
According to (\ref{p_2TR_UU}), we have $p_i - b_i\lambda_i = a_i c$ for $i \in \{1, 2\}$. (\ref{c_K=2}) can then be written as 
\begin{equation}\label{c_2_K=2}
    c = \prod_{i=1}^{2} \frac{a_i c}{a_i c + b_i \lambda_i}, 
    \vspace{-0.1cm}
\end{equation}
which has one attracting fixed point $c_L$ and one repealling fixed point $c_S$, given by
\begin{equation}\label{c_L}
    c_L = \tfrac{1}{2}\left( 1 - \tfrac{b_1 \lambda_1}{a_1} - \tfrac{b_2 \lambda_2}{a_2} + \sqrt{( 1 - \tfrac{b_1 \lambda_1}{a_1} - \tfrac{b_2 \lambda_2}{a_2})^2 - 4 \cdot \tfrac{b_1 \lambda_1}{a_1}  \cdot\tfrac{b_2 \lambda_2}{a_2}}  \right),  
\end{equation}
\begin{equation}\label{c_S}
    c_S = \tfrac{1}{2}\left( 1 - \tfrac{b_1 \lambda_1}{a_1} - \tfrac{b_2 \lambda_2}{a_2} - \sqrt{( 1 - \tfrac{b_1 \lambda_1}{a_1} - \tfrac{b_2 \lambda_2}{a_2})^2 - 4 \cdot \tfrac{b_1 \lambda_1}{a_1} \cdot \tfrac{b_2 \lambda_2}{a_2}}  \right),
    \vspace{-0.1cm}
\end{equation}
if and only if
\vspace{-0.1cm} 
\begin{equation}\label{condition_1}
    \sqrt{\tfrac{b_1 \lambda_1}{a_1} } + \sqrt{\tfrac{b_2 \lambda_2}{a_2}} < 1. 
\end{equation}
Otherwise, (\ref{c_2_K=2}) has zero positive real roots. Accordingly, (\ref{p_2TR_UU}) has two sets of roots $\bm{p}_L^{K=2} = (p_{1,L}^{K=2}, p_{2,L}^{K=2})$ and $\bm{p}_S^{K=2} = (p_{1,S}^{K=2}, p_{2,S}^{K=2})$, where 
\begin{equation}\label{p_L_S_2TR}
    p_{i,L}^{K=2} = a_i c_L + b_i \lambda_i, \quad p_{i,S}^{K=2} = a_i c_S + b_i \lambda_i,
\end{equation}
$i\in\{1,2\}$, if (\ref{condition_1}) holds. Only $\bm{p}_L^{K=2}$ is the attracting fixed point, which is the all-unsaturated steady-state point.

\subsubsection{$K$ Symmetric T-R Pairs}\label{subsubsection: symm_K_TR_p}

In this case, all $K$ T-R pairs have identical input rates, transmission probabilities, and SINR thresholds, i.e., $\lambda_i = \lambda$, $q_i = q$, and $\theta_i = \theta$ for all $i\in\mathcal{K}$. Moreover, the mean received SNR for any T-R pair is equal, i.e., $\rho_{i,j} = \rho$ for all $i,j \in \mathcal{K}$. A single-cell network with uplink power control can be regarded as an example of symmetric T-R pairs. 

With the above symmetric setting, the network has only two possible states: all-saturated or all-unsaturated. The all-saturated steady-state point $p_A$ can be obtained from (\ref{p}) as
\begin{equation}\label{p_single_symmetric_saturated}
    p_{A} = \exp\left( -\tfrac{\theta}{\rho}\right) \cdot \left(1- \tfrac{\theta}{\theta + 1}\cdot q  \right)^{K-1} \approx \exp\left( -\tfrac{\theta}{\rho} - \tfrac{K \theta}{\theta + 1}\cdot q \right),
\end{equation}
by applying $(1-x)^n \approx \exp(-nx) $ for large $n$, which is consistent with Eq. (12) given in \cite{MaximumsumrateCapture_Li}. 
When the network is all-unsaturated, (\ref{p}) becomes
\begin{equation}\label{p_single_symmetric_unsaturated}
    p = \exp\left( -\tfrac{\theta}{\rho}\right) \cdot \left(1- \tfrac{\theta}{\theta + 1}\cdot \tfrac{\lambda}{p}  \right)^{K-1} \approx \exp\left( -\tfrac{\theta}{\rho} - \tfrac{K \theta}{\theta + 1}\cdot \tfrac{\lambda}{p} \right),
\end{equation}
which has one attracting fixed point $p_L$ and one repelling fixed point $p_S$, given by
\vspace{-0.1cm}
\begin{equation}\label{p_single_symmetric_unsaturated_solution}
    p_L = \tfrac{K\theta \lambda}{-(\theta + 1) \mathbb{W}_0\left( -\tfrac{K\theta \lambda}{\theta + 1} \exp\left(\tfrac{\theta}{\rho}\right) \right)}, 
    \; p_S  = \tfrac{K\theta \lambda}{-(\theta + 1) \mathbb{W}_{-1}\left( -\tfrac{K\theta \lambda}{\theta + 1} \exp\left(\tfrac{\theta}{\rho}\right) \right)},
\end{equation}
if and only if 
\vspace{-0.1cm}
\begin{equation}\label{symm_condition_1}
    0<\lambda< \tfrac{\theta + 1}{K\theta} \exp\left(-1-\tfrac{\theta}{\rho}\right),
\end{equation}
where $-1 \leq \mathbb{W}_0(z) < 0$ and $\mathbb{W}_{-1}(z)\leq -1$ for $-e^{-1}\leq z<0$ are two branches of the Lambert W function. $p_L$ is the all-unsaturated steady-state point.

\vspace{-0.3cm} 
\section{Transmission Control for Achieving Stability}\label{section:transmission control}

It has been shown in Section \ref{subsection: transmission control_system model} that to stabilize the network, the transmission probabilities of T-R pairs should be selected from the all-unsaturated region $S_{U^K}(\bm{q},\bm{\lambda})$ defined in (\ref{define_UU}), such that the network can operate at the all-unsaturated steady-state point $\bm{p}_L$. In the following, we will first demonstrate how to calculate the all-unsaturated region $S_{U^K}(\bm{q},\bm{\lambda})$ of transmission probabilities in the general scenario, and then present the explicit expressions of $S_{U^K}(\bm{q},\bm{\lambda})$ in two special cases, i.e., the two T-R pairs and $K$ symmetric T-R pairs. 

\vspace{-0.2cm}
\subsection{All-Unsaturated Region $S_{U^K}(\bm{q}, \bm{\lambda})$} \label{subsection:SUU_general}
\vspace{-0.1cm}

\setlength{\textfloatsep}{8pt}
\begin{algorithm}[t!]
    \caption{Calculation of the All-unsaturated Region $S_{U^K}(\bm{q},\bm{\lambda})$}\label{alg:findUU}
    \begin{algorithmic}[1]
      \Require Input rate $\lambda_i$ for all $i\in\mathcal{K}$, SINR threshold $\theta_i$ for all $i\in\mathcal{L}$, and mean received SNR $\rho_{i,j}$ for all $i\in\mathcal{K}$, $j\in\mathcal{L}$.
      \renewcommand{\algorithmicrequire}{\textbf{Function:}} 
      \Require $\bm{g}_1(\bm{q}) = -\bm{q}$, $\bm{g}_2(\bm{q}) = \bm{q}$, $\bm{h}_1 (\bm{q}) = \bm{\lambda} - \bm{p} \circ \bm{q}$,\footnotemark[5] $\bm{h}_2 (\bm{q}) = \bm{q} - \bm{1}$, $\bm{h}_3(\bm{q}) = -\bm{q}$. 
      \State $\mathbf{Q}_1 \gets \text{gamultiobj}(\bm{g}_1(\bm{q}), \bm{h}_1(\bm{q}), \bm{h}_2(\bm{q}), \bm{h}_3(\bm{q}))$.\footnotemark[6]
      \State $\mathbf{Q}_2 \gets \text{gamultiobj}(\bm{g}_2(\bm{q}), \bm{h}_1(\bm{q}), \bm{h}_2(\bm{q}), \bm{h}_3(\bm{q}))$.
      \State  $S_{U^K}(\bm{q}, \bm{\lambda}) \gets \{ \bm{q} : \forall \bm{q}^u \in \mathbf{Q}_{1}, \forall \bm{q}^l \in \mathbf{Q}_{2}, \bm{q}^l \leq \bm{q} \leq \bm{q}^u \} $.
      \Ensure The all-unsaturated region $S_{U^K}(\bm{q}, \bm{\lambda})$.
    \end{algorithmic}
\end{algorithm}

\footnotetext[5]{The steady-state probabilities of successful transmission of HOL packets $\bm{p}$ for given input rates $\bm{\lambda}$ and transmission probabilities $\bm{q}$ are calculated through Algorithm \ref{alg:iter}.}

\footnotetext[6]{Here $\bm{g}_1(\bm{q})$ is the objective function. $\bm{h}_1(\bm{q})$, $\bm{h}_2(\bm{q})$, and $\bm{h}_3(\bm{q})$ are the constraint functions. The settings of the gamultiobj function are given as follows: the population size is 500, constraint tolerance is $10^{-7}$, crossover function is single-point, function tolerance is $10^{-5}$, maximum number of generations is 1000, and all the other options are set to default.  }

\begin{figure*}
    \captionsetup[subfigure]{justification=centering}
    \centering
    \subfloat[]{
      \includegraphics[width=1.55in,height=1.35in]{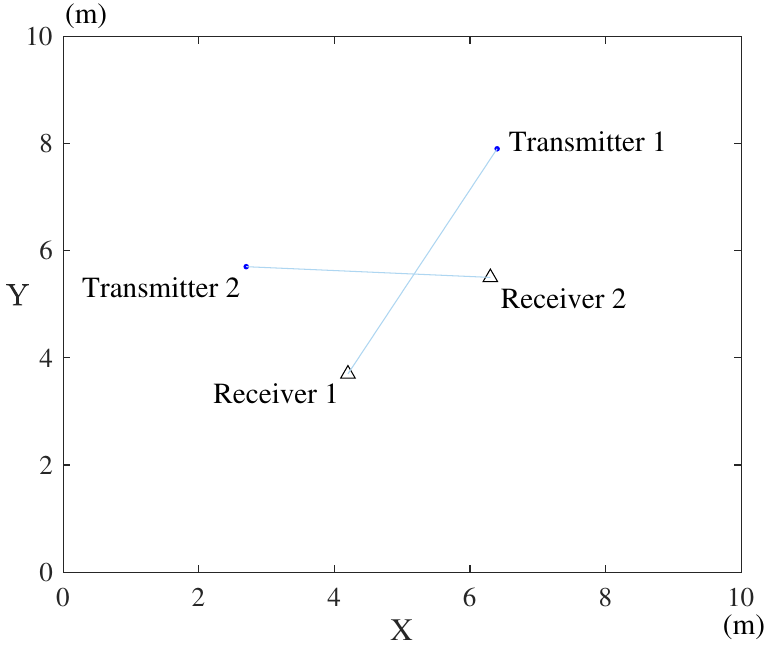}
    \label{subfig:topology_AdHoc_PPPTR_Txnum2_case1}}
    \subfloat[]{
        \includegraphics[width=1.6in,height=1.3in]{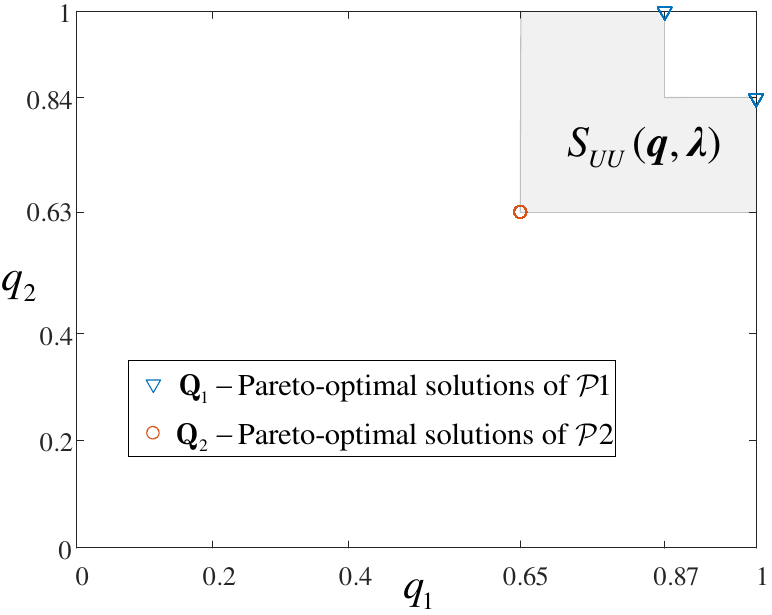}
      \label{subfig:ga_UU_AdHoc_PPPTR_Txnum2_case1_Pt-66dBm_-69dBm_thres-5dB_-7dB_irate0.2_0.27}}
      \subfloat[]{
        \includegraphics[width=1.55in,height=1.35in]{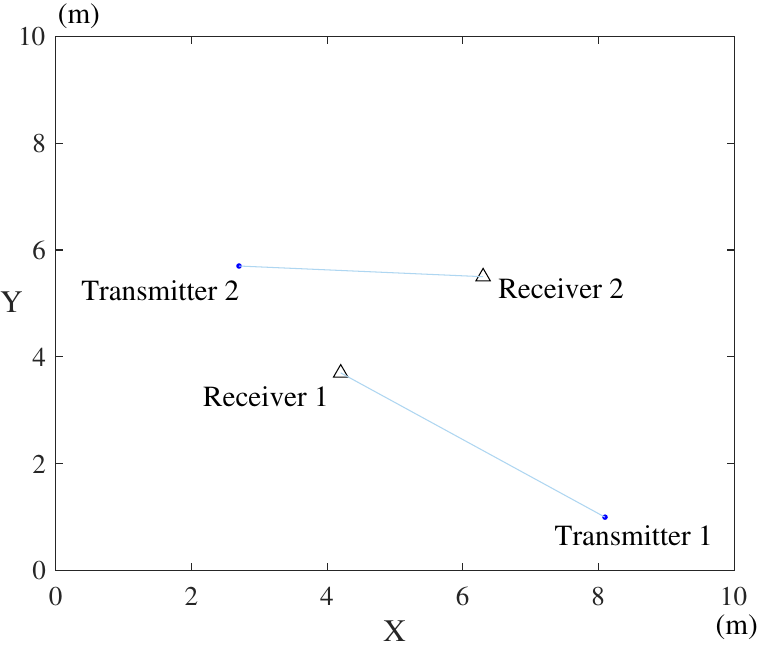}
      \label{subfig:topology_AdHoc_PPPTR_Txnum2_case2}}
      \subfloat[]{
        \includegraphics[width=1.6in,height=1.3in]{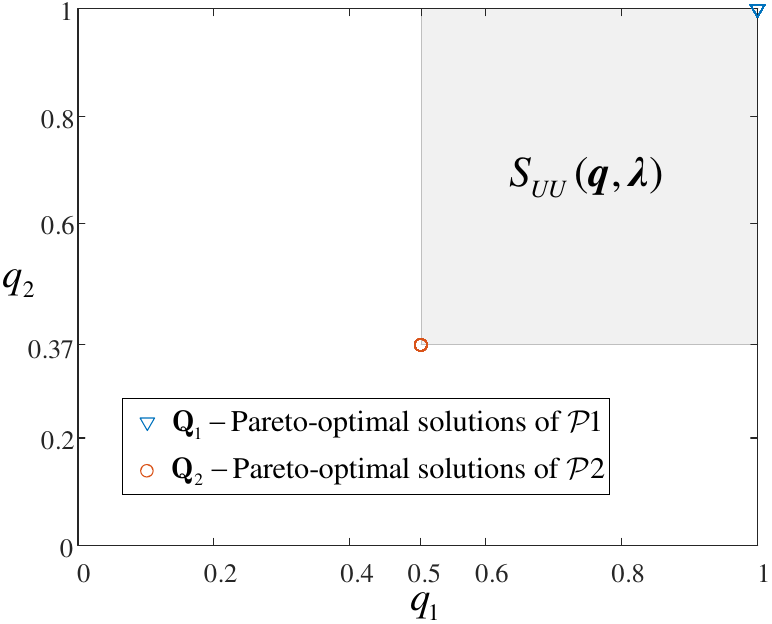}
      \label{subfig:ga_UU_AdHoc_PPPTR_Txnum2_case2_Pt-66dBm_-69dBm_thres-5dB_-7dB_irate0.2_0.27}}
      \vspace{-0.2cm}
    \caption{(a)(c) Topologies of two T-R pairs. $\rho_{1,1} = -3$ dB. $\rho_{2,2} = -1.3$ dB. $\rho_{2,1} = 5.1$ dB. (a) $\rho_{1,2} = 8.8$ dB. (c) $\rho_{1,2} = -3.4$ dB.  (b)(d) All-unsaturated region $S_{UU} (\bm{q}, \bm{\lambda})$ of transmission probabilities $\bm{q}$ obtained via Algorithm \ref{alg:findUU} with topologies given in (a) and (c), respectively. $\theta_1 = -5$ dB. $ \theta_2 = -7$ dB. $\lambda_1 = 0.2$. $\lambda_2 = 0.27$. }
    \label{fig:topology_AdHoc_PPPTR_Txnum2_case1_ga_SUU}
    \vspace{-0.65cm}
\end{figure*}

\begin{figure*}[!t]
    \captionsetup[subfigure]{justification=centering}
    \centering
    \subfloat[]{
        \includegraphics[width=1.55in,height=1.33in]{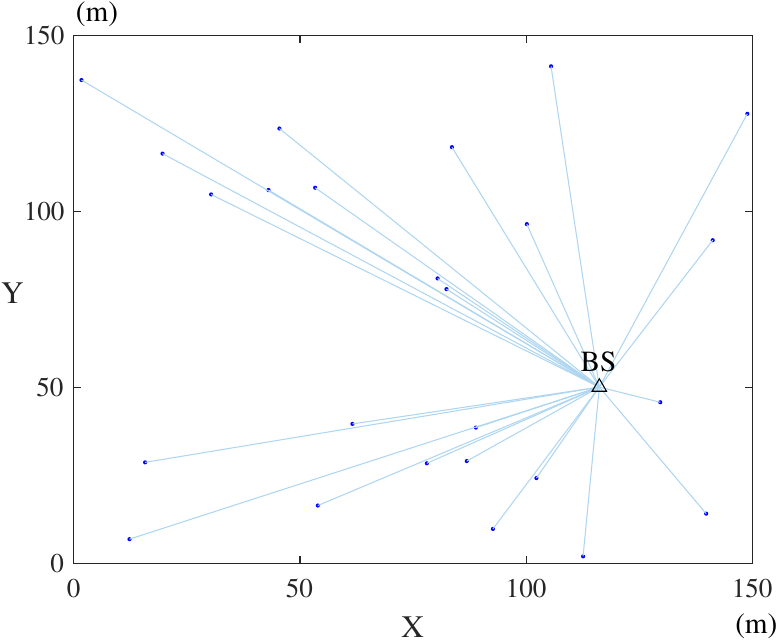}
    \label{subfig:topology_singleCell_PPP_v5_K25}} 
    \subfloat[]{
      \includegraphics[width=1.65in,height=1.28in]{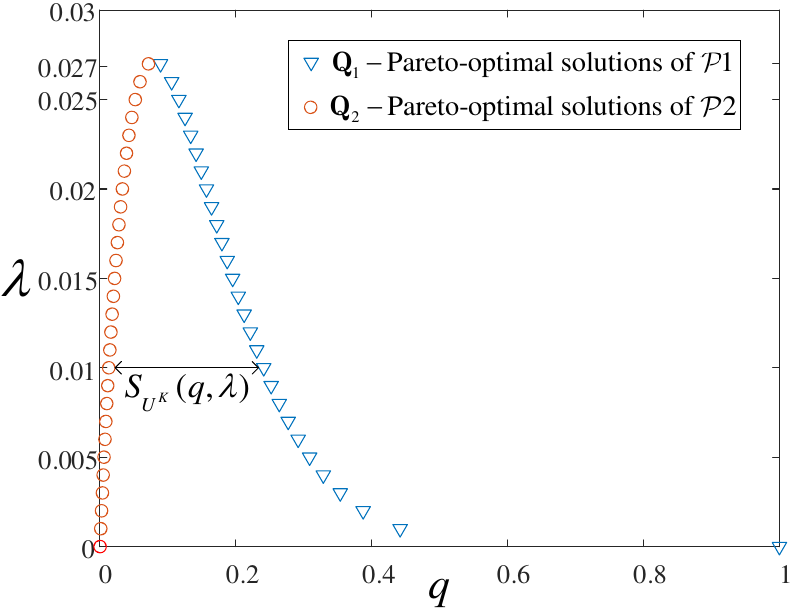}
    \label{subfig:ga_SUU_symm_vs_irate_K25_v5_rho10dB_thres0dB}}
    \subfloat[]{
        \includegraphics[width=1.55in,height=1.33in]{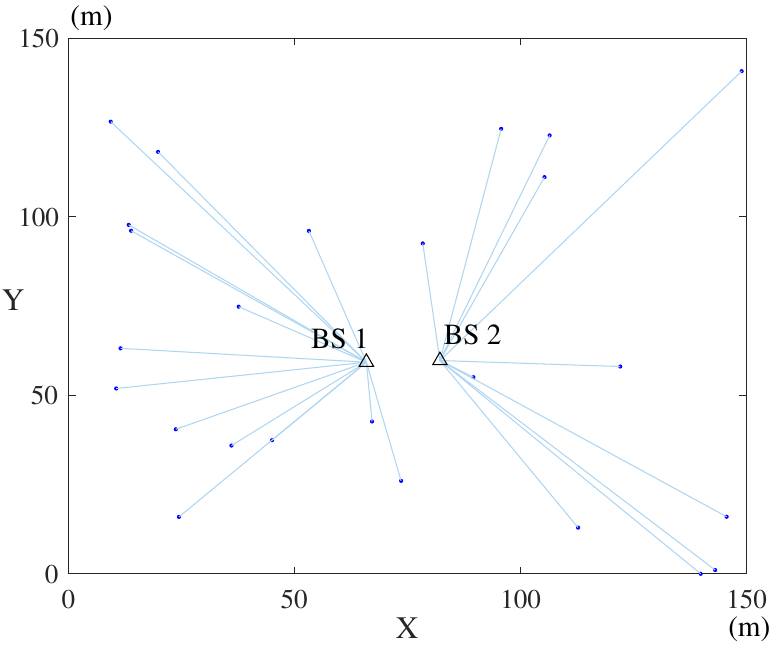}
    \label{subfig:topology_2Cell_PPP_v3}}
    \subfloat[]{
      \includegraphics[width=1.65in,height=1.28in]{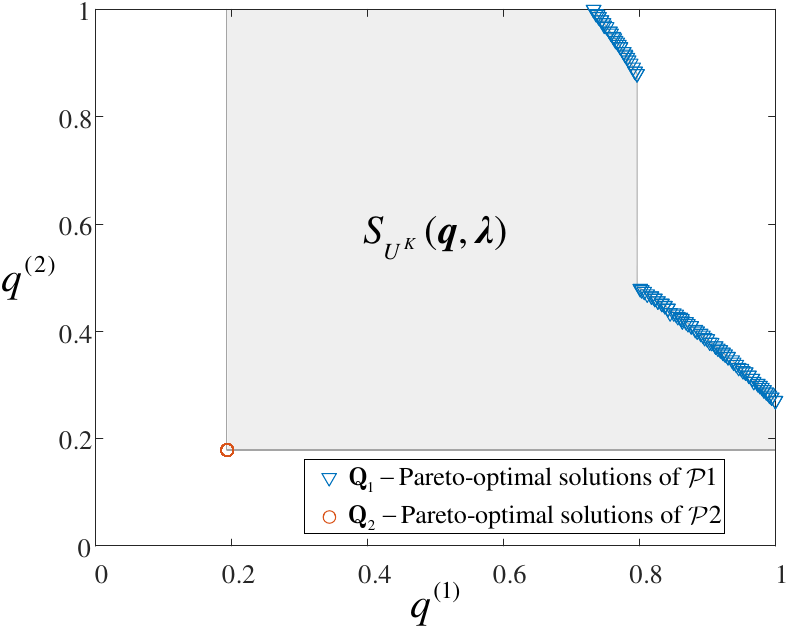}
    \label{subfig:ga_UU_2Cell_PPP_v3_rho0dB_irate0.1_thres-8dB}}
    \vspace{-0.2cm}
    \caption{(a) A single-cell network and (c) a two-cell network, where locations of transmitters are generated according to PPP with density $\xi = 10^{-3} \text{ m}^{-2}$ in $[0, 150 \text{ m}]^2$, as represented by dots. The location of each BS is randomly generated in $[0, 150 \text{ m}]^2$, as represented by the triangle, and each transmitter is associated with its closest BS. (b) All-unsaturated region $S_{U^{K}}(q, \lambda)$ of transmission probability $q$ obtained via Algorithm \ref{alg:findUU} for the single-cell network given in (a) with symmetric setting. $\rho = 10$ dB. $\theta = 0$ dB. (d) All-unsaturated region $S_{U^K}(\bm{q}, \bm{\lambda})$ of transmission probabilities $\bm{q}$ obtained via Algorithm \ref{alg:findUU} for the two-cell network given in (c). $\rho_{i,i^{\ast}} = 0$ dB, $\lambda_i = 0.1$, $i\in \mathcal{K}$. $\theta_1 = \theta_2 = -8$ dB.  $\alpha = 4$. }
    \label{fig:topology_2Cell_PPP_v3_UU}
    \vspace{-0.4cm}
\end{figure*}

According to (\ref{define_UU}), for given input rates $\bm{\lambda}$, to determine the all-unsaturated region $S_{U^K}(\bm{q},\bm{\lambda})$, we need to find all possible values of transmission probabilities $\bm{q}$, with which the service rates $\bm{\mu}$ are larger than the input rates $\bm{\lambda}$, or equivalently, the network operates at the all-unsaturated steady-state point $\bm{p}_L$. Therefore, the all-unsaturated region $S_{U^K}(\bm{q},\bm{\lambda})$ can be determined based on the Pareto-optimal solutions of the following multi-objective problems:
\vspace{-0.1cm}  
\begin{align}
    \mathcal{P}1:  \quad     \max \quad   & \bm{q}, \label{P1_objfun} \\ 
        \text{s.t.} \quad   & \bm{0} < \bm{q} \leq \bm{1},  \label{P1_linconstraint}\\
                          & \bm{p} = \bm{p}_L,  \label{P1_nonlinconstraint}
    \vspace{-0.3cm}
\end{align}
and
\vspace{-0.3cm}
\begin{align}
    \mathcal{P}2:  \quad    \min \quad   & \bm{q}, \label{P2_objfun} \\
      \text{s.t.} \quad   & \bm{0} < \bm{q} \leq \bm{1}, \label{P2_linconstraint} \\
                         &  \bm{p} = \bm{p}_L.  \label{P2_nonlinconstraint}
 \vspace{-0.2cm}
\end{align} 
Note that according to (\ref{service_rate}), for $\bm{p} = \bm{p}_L$, $\bm{p} \circ \bm{q} > \bm{\lambda}$ should hold, where the network steady-state point $\bm{p}$ for given $\bm{\lambda}$ and $\bm{q}$ can be obtained through Algorithm \ref{alg:iter}. Algorithm \ref{alg:findUU} summarizes the main steps of calculating the all-unsaturated region $S_{U^K}(\bm{q},\bm{\lambda})$, where the elitist non-dominated sorting genetic algorithm (NSGA) \cite{MultiObjective_Deb} implemented by using the gamultiobj function in MATLAB is adopted to find the Pareto-optimal solutions of $\mathcal{P}1$ and $\mathcal{P}2$.

For illustration, let us consider two T-R pairs with a random topology shown in Fig. \ref{subfig:topology_AdHoc_PPPTR_Txnum2_case1}. The Pareto-optimal solutions $\mathbf{Q}_1$ and $\mathbf{Q}_2$ can be obtained from Algorithm \ref{alg:findUU} as $\mathbf{Q}_1 = \{(0.87, 1), (1, 0.84)\}$ and $\mathbf{Q}_2 = \{(0.65, 0.63)\}$, and the corresponding all-unsaturated region $S_{UU} (\bm{q}, \bm{\lambda})$ is shown in Fig. \ref{subfig:ga_UU_AdHoc_PPPTR_Txnum2_case1_Pt-66dBm_-69dBm_thres-5dB_-7dB_irate0.2_0.27}. When the distance between Transmitter 1 and Receiver 2 increases, as Fig. \ref{subfig:topology_AdHoc_PPPTR_Txnum2_case2} shows, the interference between T-R pairs becomes lower comparing to the one in Fig. \ref{subfig:topology_AdHoc_PPPTR_Txnum2_case1}. It can be seen from Fig. \ref{subfig:ga_UU_AdHoc_PPPTR_Txnum2_case2_Pt-66dBm_-69dBm_thres-5dB_-7dB_irate0.2_0.27} that with $\mathbf{Q}_1$ and $\mathbf{Q}_2$ obtained as $\mathbf{Q}_1 = \{(1, 1)\}$ and $\mathbf{Q}_2 = \{(0.5, 0.37)\}$, the all-unsaturated region $S_{UU} (\bm{q}, \bm{\lambda})$ is significantly larger than that in Fig. \ref{subfig:ga_UU_AdHoc_PPPTR_Txnum2_case1_Pt-66dBm_-69dBm_thres-5dB_-7dB_irate0.2_0.27} owing to lower interference between the T-R pairs.

For cellular networks, in Fig. \ref{subfig:topology_singleCell_PPP_v5_K25}, we further consider a single cell with symmetric setting, that is, $q_i = q$, $\lambda_i = \lambda$, $\rho_{i,i^{\ast}} = \rho$, and $\theta_{i^{\ast}} = \theta$ for all $i\in\mathcal{K}$. 
The all-unsaturated region $S_{U^{K}}(q, \lambda)$ of transmission probability $q$ can be obtained from Algorithm \ref{alg:findUU} for different values of input rate $\lambda$, as Fig. \ref{subfig:ga_SUU_symm_vs_irate_K25_v5_rho10dB_thres0dB} shows. For the two-cell network shown in Fig. \ref{subfig:topology_2Cell_PPP_v3}, let $q_i = q^{(l)}$ and $\lambda_i = \lambda^{(l)}$ be the transmission probability and input rate of each transmitter in Cell $l$, respectively, $i\in \mathcal{K}_l^C$, where $\mathcal{K}_l^C$ is the set of all transmitters associated with BS $l$ for $l\in \{1,2\}$ and $\mathcal{K}_1^C \cup \mathcal{K}_2^C = \mathcal{K}$.     
The corresponding all-unsaturated region  $S_{U^K}(\bm{q}, \bm{\lambda})$ obtained from Algorithm \ref{alg:findUU} is illustrated in Fig. \ref{subfig:ga_UU_2Cell_PPP_v3_rho0dB_irate0.1_thres-8dB}.

In general, the all-unsaturated region $S_{U^K} (\bm{q}, \bm{\lambda})$ has to be numerically calculated based on Algorithm \ref{alg:findUU}. In the following, we will demonstrate that in two special cases, i.e., two T-R pairs and $K$ symmetric T-R pairs, the explicit expressions of $S_{U^K} (\bm{q}, \bm{\lambda})$ can be obtained. 

\subsection{Special Cases} \label{subsection:special cases_SUU}

As we have mentioned in Section \ref{subsection:SUU_general}, to achieve stability, the network needs to operate at the all-unsaturated steady-state point $\bm{p}_L$. In Section \ref{subsection:special cases_p}, we have derived the all-unsaturated steady-state points for the two special cases, i.e., two T-R pairs and $K$ symmetric T-R pairs. The following theorems further present the all-unsaturated region $S_{U^K}(\bm{q},\bm{\lambda})$ in the two special cases.

\subsubsection{Two T-R Pairs}\label{subsubsection: 2 T-R pairs_SUU}
With $K=L=2$, the all-unsaturated steady-state point $\bm{p}_L^{K=2}$ has been given in (\ref{p_L_S_2TR}). Theorem \ref{theorem: SUU_2TR} presents the corresponding all-unsaturated region $S_{UU}(\bm{q},\bm{\lambda})$.
\begin{theorem}\label{theorem: SUU_2TR}
    The all-unsaturated region $S_{UU}(\bm{q},\bm{\lambda})$ for $K=L=2$ is 
    \vspace{-0.2cm}
    \begin{multline} \label{SUU_2TR}
        S_{UU}(\bm{q},\bm{\lambda}) = \Biggl\{ \bm{q}: \underset{i\neq j}{\bigcup_{i,j\in\{1,2\}}} \Biggl\{  \tfrac{\lambda_i}{p_{i,L}^{K=2}} < q_i < \min\left(\tfrac{\lambda_i}{p_{i,S}^{K=2}}, 1\right), \\
        \tfrac{\lambda_j}{p_{j,L}^{K=2}} < q_j \leq 1 \Biggr\}   \Biggr\}. 
        \vspace{-0.3cm}
    \end{multline} 
\end{theorem}
\begin{proof}
    See Appendix A.
\end{proof}

\begin{figure}
    \centering
    \subfloat[]{
        \includegraphics[width=1.5in,height=1.4in]{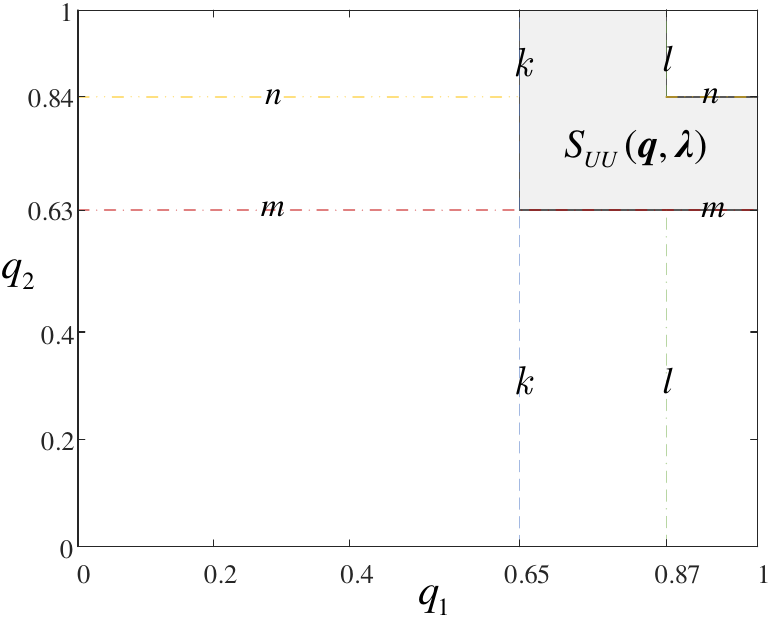}
      \label{subfig:SUU_AdHoc_PPPTR_Txnum2_case1_Pt-66_-69dBm_thres-5dB_-7dB_irate0.2_0.27}}
    \subfloat[]{
        \includegraphics[width=1.5in,height=1.4in]{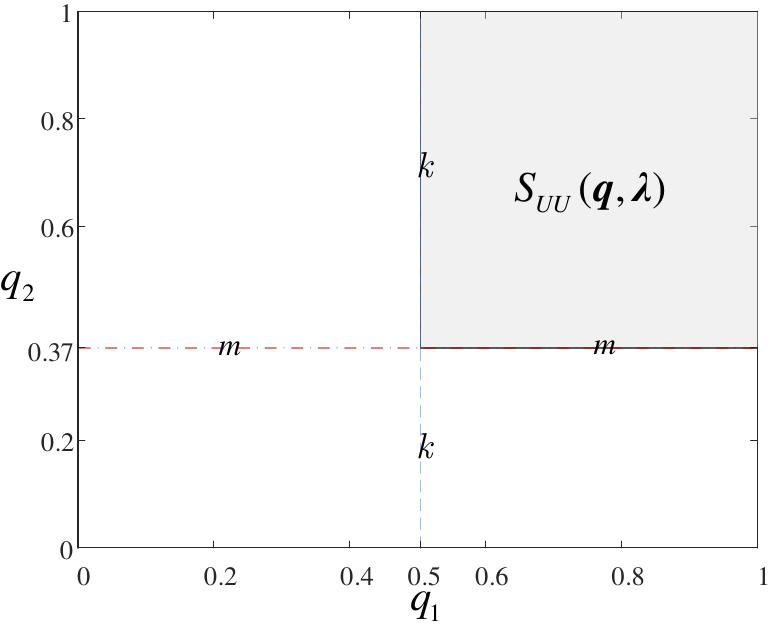}
      \label{subfig:SUU_AdHoc_PPPTR_Txnum2_case2_Pt-66_-69dBm_thres-5dB_-7dB_irate0.2_0.27}}
    \subfloat{
        \includegraphics[width=0.55in,height=1.4in]{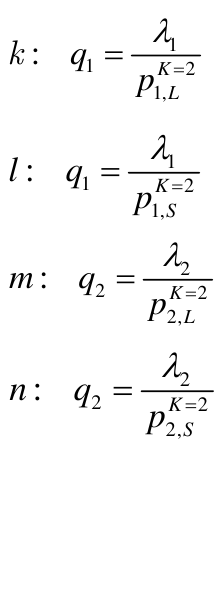}
      \label{subfig:SUU_2TR_label}}
    \vspace{-0.2cm}
    \caption{All-unsaturated region $S_{UU}(\bm{q},\bm{\lambda})$ of transmission probabilities $\bm{q}$ obtained from Theorem \ref{theorem: SUU_2TR} for the two T-R pairs with the topology given in (a) Fig. \ref{subfig:topology_AdHoc_PPPTR_Txnum2_case1} and (b) Fig. \ref{subfig:topology_AdHoc_PPPTR_Txnum2_case2}.  $\theta_1 = -5$ dB. $ \theta_2 = -7$ dB. $\lambda_1 = 0.2$. $\lambda_2 = 0.27$.}
    \label{fig:SUU_AdHoc_PPPTR_Txnum2_Pt-66_-69dBm_thres-5dB_-7dB_irate0.2_0.27}
\end{figure}

For the two T-R pairs with the topology given in Fig. \ref{subfig:topology_AdHoc_PPPTR_Txnum2_case1} and Fig. \ref{subfig:topology_AdHoc_PPPTR_Txnum2_case2}, Fig. \ref{subfig:SUU_AdHoc_PPPTR_Txnum2_case1_Pt-66_-69dBm_thres-5dB_-7dB_irate0.2_0.27} and Fig. \ref{subfig:SUU_AdHoc_PPPTR_Txnum2_case2_Pt-66_-69dBm_thres-5dB_-7dB_irate0.2_0.27} illustrate the corresponding all-unsaturated region $S_{UU}(\bm{q},\bm{\lambda})$ obtained from Theorem \ref{theorem: SUU_2TR}. By comparing Figs. \ref{subfig:SUU_AdHoc_PPPTR_Txnum2_case1_Pt-66_-69dBm_thres-5dB_-7dB_irate0.2_0.27} and \ref{subfig:SUU_AdHoc_PPPTR_Txnum2_case2_Pt-66_-69dBm_thres-5dB_-7dB_irate0.2_0.27} with Figs. \ref{subfig:ga_UU_AdHoc_PPPTR_Txnum2_case1_Pt-66dBm_-69dBm_thres-5dB_-7dB_irate0.2_0.27} and \ref{subfig:ga_UU_AdHoc_PPPTR_Txnum2_case2_Pt-66dBm_-69dBm_thres-5dB_-7dB_irate0.2_0.27}, respectively, it can be seen that Algorithm \ref{alg:findUU} and Theorem \ref{theorem: SUU_2TR} lead to the same result. By expressing $S_{UU}(\bm{q},\bm{\lambda})$ as an explicit function, however, the effects of key parameters, such as input rates $\bm{\lambda}$, can be clearly observed. For instance, by noting that $\tfrac{\lambda_i}{p_{i,L}^{K=2}} = \tfrac{1}{b_i}\left(1 - c_L - \tfrac{b_j \lambda_j}{a_j}\right)$ and $\tfrac{\lambda_i}{p_{i,S}^{K=2}} = \tfrac{1}{b_i}\left(1 - \tfrac{p_{j, S}^{K=2}}{a_j}\right) = \tfrac{1}{b_i}\left(1 - c_S - \tfrac{b_j \lambda_j}{a_j}\right)$, where $c_L$ decreases and $c_S$ increases as $\lambda_i$ increases, we can see from (\ref{SUU_2TR}) that the all-unsaturated region $S_{UU}(\bm{q},\bm{\lambda})$ shrinks as the input rate of any transmitter increases, which means it is more difficult to stabilize the two T-R pairs.

\subsubsection{$K$ Symmetric T-R Pairs}\label{subsubsection: symmetric_SUU}
With $\lambda_i = \lambda$, $q_i = q$, $\theta_i = \theta$ for all $i\in\mathcal{K}$, and $\rho_{i,j} = \rho$ for all $i,j\in\mathcal{K}$, the all-unsaturated steady-state point $p_L$ is given in (\ref{p_single_symmetric_unsaturated_solution}). Theorem \ref{theorem: SUU_symmetric} presents the all-unsaturated region $S_{U^K}(q,\lambda)$ for $K$ symmetric T-R pairs.

\begin{theorem}\label{theorem: SUU_symmetric}
    The all-unsaturated region $S_{U^K}(q,\lambda)$ for $K$ symmetric T-R pairs is 
    \vspace{-0.2cm}
    \begin{multline}\label{SUU_symmetric}
        S_{U^K}(q,\lambda) = \Bigg\{q: \tfrac{-(\theta + 1) \mathbb{W}_0\left( -\tfrac{K\theta \lambda}{\theta + 1} \exp\left(\tfrac{\theta}{\rho}\right) \right)}{K\theta} < q \\
         < \min\left( \tfrac{-(\theta + 1) \mathbb{W}_{-1}\left( -\tfrac{K\theta \lambda}{\theta + 1} \exp\left(\tfrac{\theta}{\rho}\right) \right)}{K\theta}, 1  \right) \Bigg\}.
         \vspace{-0.2cm}
    \end{multline}
\end{theorem}
\begin{proof}
    See Appendix B.
\end{proof}

Consider the symmetric single-cell network given in Fig. \ref{subfig:topology_singleCell_PPP_v5_K25}. With $K=25$, the corresponding all-unsaturated region $S_{U^{25}}(q,\lambda)$ can be obtained from Theorem \ref{theorem: SUU_symmetric} and illustrated in Fig. \ref{fig:SUU_symm_vs_irate_K40_25_rho10dB_thres0dB}. By comparing Fig. \ref{fig:SUU_symm_vs_irate_K40_25_rho10dB_thres0dB} with Fig. \ref{subfig:ga_SUU_symm_vs_irate_K25_v5_rho10dB_thres0dB}, it can be observed that Theorem \ref{theorem: SUU_symmetric} leads to the same result as the one obtained through Algorithm \ref{alg:findUU}. Yet with the explicit expression given in (\ref{SUU_symmetric}), it can be easily seen that the all-unsaturated region $S_{U^{K}}(q,\lambda)$ is closely dependent on the number of transmitters $K$ and input rate $\lambda$. Fig. \ref{fig:SUU_symm_vs_irate_K40_25_rho10dB_thres0dB} also corroborates that as 
$K$ or $\lambda$ increases, $S_{U^{K}}(q,\lambda)$ quickly shrinks, indicating that the transmission probability $q$ should be more carefully controlled for achieving stability as the contention becomes severer. Moreover, it can be observed from Fig. \ref{fig:SUU_symm_vs_irate_K40_25_rho10dB_thres0dB} and Fig. \ref{subfig:ga_SUU_symm_vs_irate_K25_v5_rho10dB_thres0dB} that if $\lambda$ is too large, e.g., $\lambda>0.027$ for $K = 25$, the network cannot be stabilized as $S_{U^{K}}(q,\lambda)$ becomes an empty set. The upper-bound of input rate is reduced to $0.017$ if the number of transmitters is increased to $K = 40$.

\begin{figure}[!t]
    \centering
    \includegraphics[width=2.3in,height=1.7in]{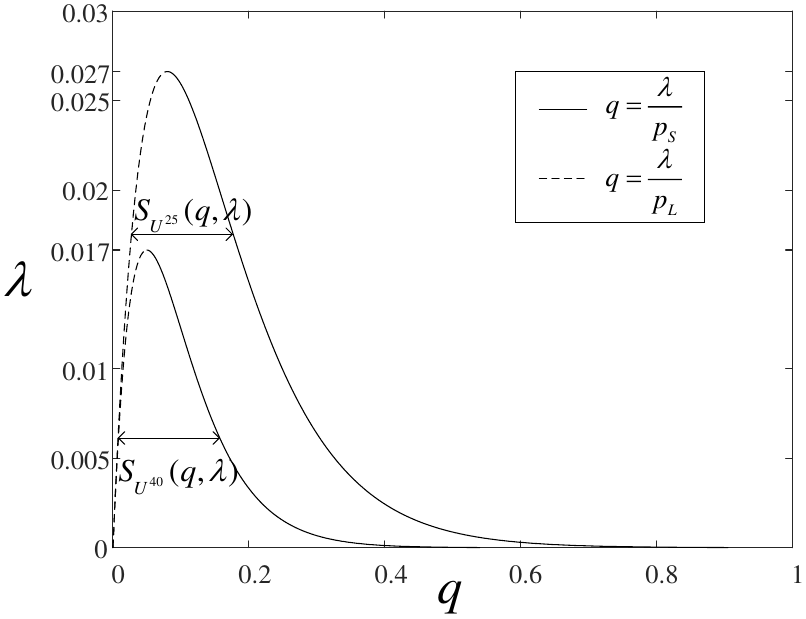}
    \vspace{-0.3cm}
    \caption{All-unsaturated region $S_{U^K}(q,\lambda)$ of transmission probability $q$ obtained from Theorem \ref{theorem: SUU_symmetric} for $K$ symmetric T-R pairs with the topology given in Fig. \ref{subfig:topology_singleCell_PPP_v5_K25}. $K \in \{25, 40\}$. $\rho = 10$ dB. $\theta = 0$ dB.}
    \label{fig:SUU_symm_vs_irate_K40_25_rho10dB_thres0dB}
\end{figure}

\vspace{-0.3cm}
\section{Maximum Input Rates for Achieving Stability}\label{section: stability region}

We can see from the above that there exists a certain upper-bound of input rates $\bm{\lambda}$, only below which the all-unsaturated region is not empty and the network can be stabilized by properly choosing the transmission probabilities. As we have mentioned in Section \ref{subsection: stability region_system model}, the stability region $S_Q(\bm{\lambda})$ dictates such an upper-bound of input rates. In this section, we will characterize the stability region $S_Q(\bm{\lambda})$ of input rates $\bm{\lambda}$.

\vspace{-0.27cm}
\subsection{Stability Region $S_Q(\bm{\lambda})$}\label{subsection:SQ_general}
\vspace{-0.03cm}

According to (\ref{define_SQ}), the all-unsaturated region $S_{U^K}(\bm{q},\bm{\lambda})$ of transmission probabilities $\bm{q}$ exists if and only if the input rates $\bm{\lambda}$ are within the stability region $S_Q(\bm{\lambda})$. To determine $S_Q(\bm{\lambda})$, we can formulate the following multi-objective problem to find the maximum input rates below which $S_{U^K}(\bm{q},\bm{\lambda})$ is not empty:
\begin{align}
    \mathcal{P}3:  \quad     \max \quad   & \bm{\lambda}, \label{P3_objfun} \\ 
        \text{s.t.} \quad   & \bm{0} < \bm{\lambda} \leq \bm{1},  \label{P3_linconstraint}\\
                          &   S_{U^K}(\bm{q},\bm{\lambda}) \neq \varnothing .  \label{P3_nonlinconstraint}
\end{align}
Note that for given $\bm{\lambda}$, the corresponding $S_{U^K}(\bm{q},\bm{\lambda})$ can be obtained from Algorithm \ref{alg:findUU}. Algorithm \ref{alg:findSQ} summarizes the main steps of calculating the stability region $S_Q(\bm{\lambda})$. 

For illustration, three examples are presented in Fig. \ref{fig:ga_stableR_exhaustive}.
Specifically, for the two T-R pairs given in Fig. \ref{subfig:topology_AdHoc_PPPTR_Txnum2_case1}, the stability region $S_{Q}(\bm{\lambda})$ is shown in Fig. \ref{subfig:ga_stableR_AdHoc_PPPTR_2Txnum_case1_Pt-66_-69dBm_thres-5dB_-7dB_k_less_1}. For the single-cell network given in Fig. \ref{subfig:topology_singleCell_PPP_v5_K25} with the symmetric setting, Fig. \ref{subfig:exhaustiveSearch_stableR_symm_vs_theta_K25_v5_rho10dB} illustrates how the stability region $S_Q(\lambda)$  varies with the SINR threshold $\theta$. Fig. \ref{subfig:ga_stableR_topology_2Cell_PPP_v3_thres-8dB_rho0dB} shows the stability region $S_{Q}(\bm{\lambda})$ of the two-cell network given in Fig. \ref{subfig:topology_2Cell_PPP_v3}.

\begin{algorithm}[t!]
    \caption{Calculation of the Stability Region $S_Q(\bm{\lambda})$}\label{alg:findSQ}
    \begin{algorithmic}[1]
      \Require SINR threshold $\theta_i$ for all $i\in\mathcal{L}$ and mean received SNR $\rho_{i,j}$ for all $i\in\mathcal{K}$, $j\in\mathcal{L}$.
      \renewcommand{\algorithmicrequire}{\textbf{Function:}} 
      \Require $\bm{g}_3(\bm{\lambda}) = -\bm{\lambda}$, $\bm{h}_4(\bm{\lambda}) = 1 - |S_{U^K}(\bm{q},\bm{\lambda})|$,\footnotemark[7] $\bm{h}_5 (\bm{\lambda}) = \bm{\lambda} - 1$, $\bm{h}_6(\bm{\lambda}) = -\bm{\lambda}$.
      \State $\mathbf{Q}_3 \gets \text{gamultiobj}(\bm{g}_3(\bm{\lambda}), \bm{h}_4(\bm{\lambda}), \bm{h}_5(\bm{\lambda}), \bm{h}_6(\bm{\lambda}))$.
      \State  $S_Q(\bm{\lambda}) \gets \{ \bm{\lambda} : \forall \bm{\lambda}^{\ast} \in \mathbf{Q}_{3}, \bm{0} < \bm{\lambda} \leq \bm{\lambda}^{\ast} \} $.
      \Ensure The stability region $S_Q(\bm{\lambda})$.
    \end{algorithmic}
\end{algorithm}

\footnotetext[7]{The all-unsaturated region $S_{U^K}(\bm{q},\bm{\lambda})$ for given input rates $\bm{\lambda}$ is calculated from Algorithm \ref{alg:findUU}. }

\begin{figure*}[!t]
    \captionsetup[subfigure]{justification=centering}
    \centering
    \subfloat[]{
        \includegraphics[width=1.8in,height=1.5in]{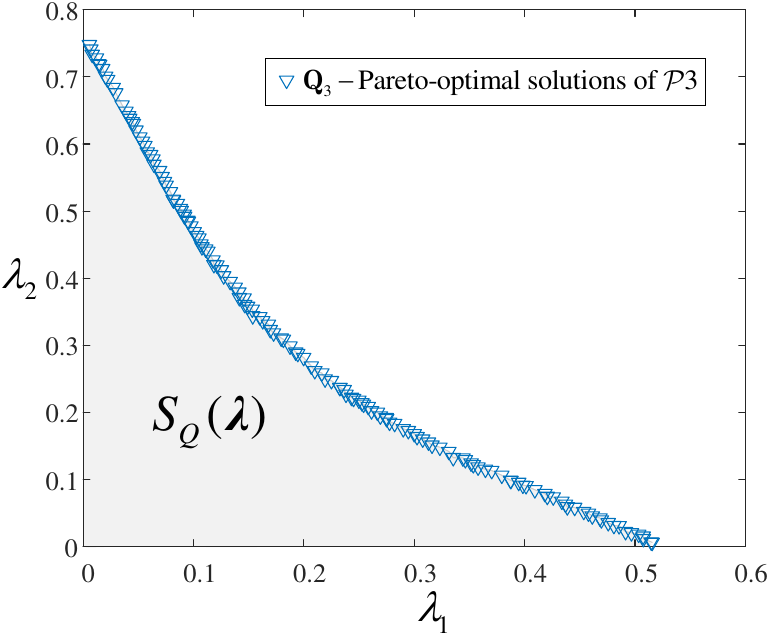}
      \label{subfig:ga_stableR_AdHoc_PPPTR_2Txnum_case1_Pt-66_-69dBm_thres-5dB_-7dB_k_less_1}}
      \hspace{0.8cm}
    \subfloat[]{
        \includegraphics[width=1.8in,height=1.5in]{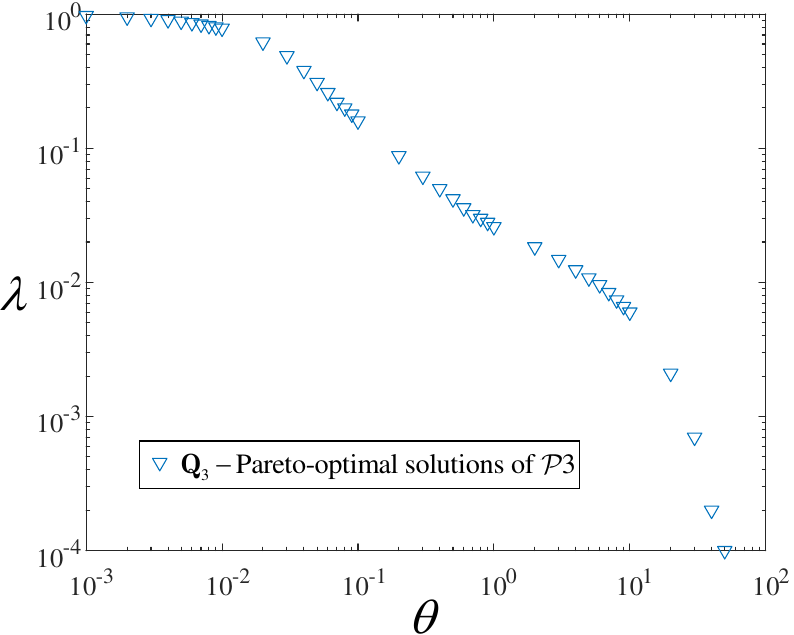}
      \label{subfig:exhaustiveSearch_stableR_symm_vs_theta_K25_v5_rho10dB}}
      \hspace{0.8cm}
    \subfloat[]{
      \includegraphics[width=1.8in,height=1.5in]{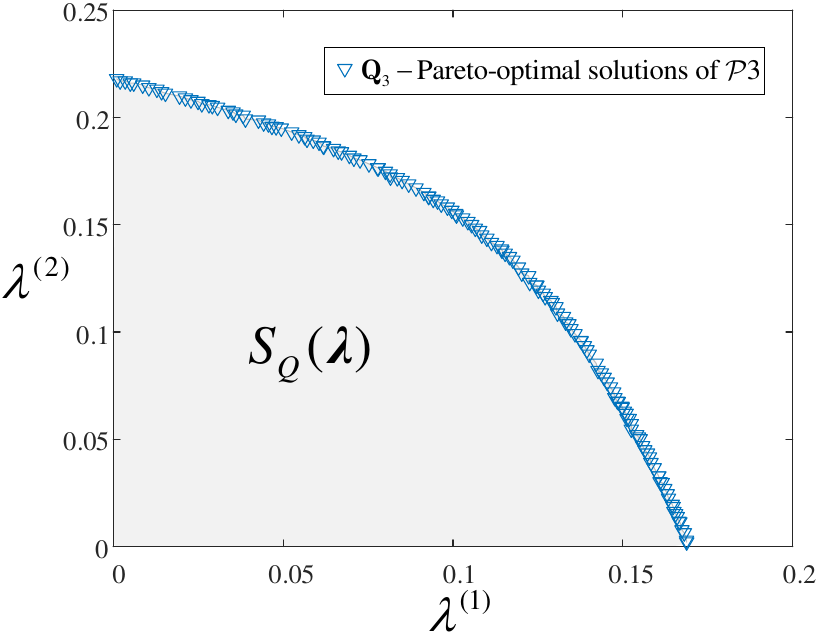}
    \label{subfig:ga_stableR_topology_2Cell_PPP_v3_thres-8dB_rho0dB}}
    \vspace{-0.1cm}
    \caption{Stability region $S_{Q}(\bm{\lambda})$ obtained through Algorithm \ref{alg:findSQ}. (a) Two T-R pairs with the topology given in Fig. \ref{subfig:topology_AdHoc_PPPTR_Txnum2_case1}.  $\theta_1 = -5$ dB. $ \theta_2 = -7$ dB. (b) A single-cell network with the topology given in Fig. \ref{subfig:topology_singleCell_PPP_v5_K25}. $\rho = 10$ dB. (c) A two-cell network with the topology given in Fig. \ref{subfig:topology_2Cell_PPP_v3}. $\theta_1 = \theta_2 = -8$ dB.}
    \label{fig:ga_stableR_exhaustive}
    \vspace{-0.2cm}
\end{figure*}

In general, the stability region $S_Q(\bm{\lambda})$ of input rates $\bm{\lambda}$ can be numerically calculated through Algorithm \ref{alg:findSQ}. 
In the following, two special cases, i.e., two T-R pairs and $K$ symmetric T-R pairs, will be considered, where explicit expressions of the stability region are obtained.

\vspace{-0.25cm}
\subsection{Special Cases} \label{subsection:special cases_stability region}
\vspace{-0.05cm}

As we have mentioned in Section \ref{subsection:SQ_general}, only when the input rates $\bm{\lambda}$ are within the stability region $S_{Q}(\bm{\lambda})$, the network can operate at the all-unsaturated steady-state point $\bm{p}_L$ by choosing transmission probabilities $\bm{q}$ from the all-unsaturated region $S_{U^K}(\bm{q},\bm{\lambda})$. 
In Sections \ref{subsection:special cases_p} and \ref{subsection:special cases_stability region}, the all-unsaturated steady-state point $\bm{p}_L$ and all-unsaturated region $S_{U^K}(\bm{q},\bm{\lambda})$ have been explicitly obtained, respectively, for the two special cases, i.e., two T-R pairs and $K$ symmetric T-R pairs. The following Theorem \ref{theorem: SQ_2TR} and Theorem \ref{theorem: SQ_symmetric} further present the stability region $S_{Q}(\bm{\lambda})$ in the two special cases.

\begin{figure*}
    \centering
    \subfloat[$\theta_1 = -5$ dB. $\theta_2 = -7$ dB.]{
      \includegraphics[width=1.8in,height=1.55in]{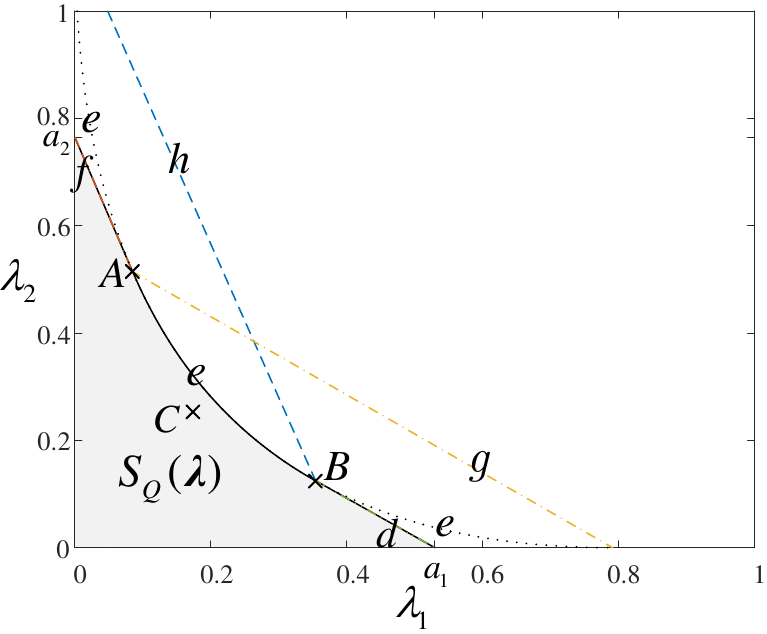}
    \label{subfig:anal_stableR_AdHoc_PPPTR_2Txnum_case1_Pt-66_-69dBm_thres-5dB_-7dB_k_less_1}}
    \subfloat[$\theta_1 = -11.2$ dB. $\theta_2 = -7$ dB. ]{
      \includegraphics[width=1.8in,height=1.55in]{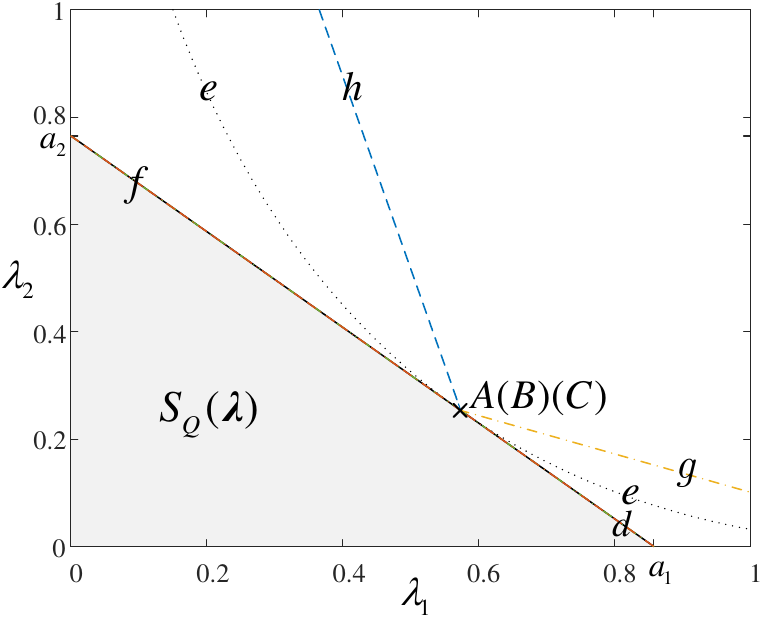}
    \label{subfig:anal_stableR_AdHoc_PPPTR_2Txnum_case1_Pt-66_-69dBm_thres-11.17dB_-7dB_k_eq_1}}
    \subfloat[$\theta_1 = -11.2$ dB. $\theta_2 = -10$ dB.]{
      \includegraphics[width=1.8in,height=1.55in]{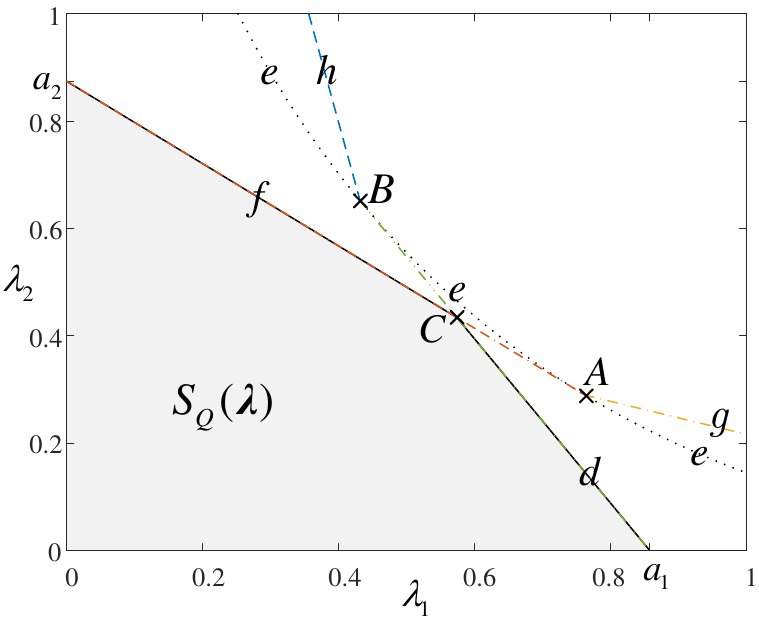}
    \label{subfig:anal_stableR_AdHoc_PPPTR_2Txnum_case1_Pt-66_-69dBm_thres-11.17dB_-10dB_k_great_1}}
    \subfloat{
        \includegraphics[width=1in,height=1.55in]{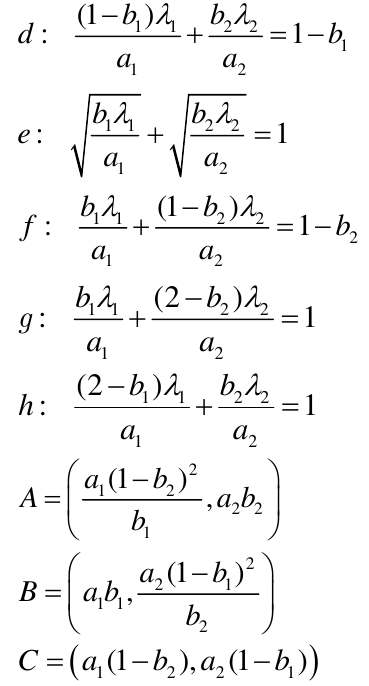}
      \label{subfig:anal_stableR_2TR_label}}
    \caption{Stability region $S_Q(\bm{\lambda})$ obtained from Theorem \ref{theorem: SQ_2TR} for the two T-R pairs with the topology given in Fig. \ref{subfig:topology_AdHoc_PPPTR_Txnum2_case1}.  (a) $0<\tfrac{\rho_{1,1}\rho_{2,2}}{\rho_{1,2}\rho_{2,1}\theta_1\theta_2} < 1$. (b)  $\tfrac{\rho_{1,1}\rho_{2,2}}{\rho_{1,2}\rho_{2,1}\theta_1\theta_2} = 1$. (c) $\tfrac{\rho_{1,1}\rho_{2,2}}{\rho_{1,2}\rho_{2,1}\theta_1\theta_2}>1$. }
    \label{fig:anal_stableR_2TR}
    \vspace{-0.5cm}
\end{figure*}

\subsubsection{Two T-R Pairs}\label{subsubsection: two TR_stability region}

With $K=L=2$, the following theorem presents the stability region $S_Q(\bm{\lambda})$ of input rates $\bm{\lambda}$.
\begin{theorem}\label{theorem: SQ_2TR}
    The stability region $S_Q(\bm{\lambda})$ for $K=L=2$ is 
    \begin{multline}\label{stability_region_2TR}
        S_Q(\bm{\lambda}) = \left\{\bm{\lambda}: \sqrt{\tfrac{b_1\lambda_1}{a_1}} + \sqrt{\tfrac{b_2\lambda_2}{a_2}} < 1\right\} \\ \bigcap \left\{\bm{\lambda}: \tfrac{(1-b_1)\lambda_1}{a_1} + \tfrac{b_2\lambda_2}{a_2}< 1 - b_1, \text{ or } \tfrac{(2-b_1)\lambda_1}{a_1} + \tfrac{b_2\lambda_2}{a_2}< 1\right\} \\
        \bigcap \left\{\bm{\lambda}: \tfrac{b_1\lambda_1}{a_1} + \tfrac{(1-b_2)\lambda_2}{a_2} < 1 - b_2, \text{ or } \tfrac{b_1\lambda_1}{a_1} + \tfrac{(2-b_2)\lambda_2}{a_2}< 1\right\},
    \end{multline}
    where $a_i$ and $b_i$ are given in (\ref{a_b}), $i\in\{1,2\}$. 
\end{theorem}
\begin{proof}
 See Appendix C.
\end{proof} 

Fig. \ref{fig:anal_stableR_2TR} illustrates the stability region $S_Q(\bm{\lambda})$ for the two T-R pairs given in Fig. \ref{subfig:topology_AdHoc_PPPTR_Txnum2_case1} with different values of SINR thresholds. As the boundary of $S_Q(\bm{\lambda})$ is determined by the following equations of $\lambda_1$ and $\lambda_2$:
\begin{align}
   & \tfrac{(1-b_1)\lambda_1}{a_1} + \tfrac{b_2\lambda_2}{a_2}  = 1 - b_1, \quad \sqrt{\tfrac{b_1\lambda_1}{a_1}} + \sqrt{\tfrac{b_2\lambda_2}{a_2}} = 1, \nonumber\\
   & \tfrac{b_1\lambda_1}{a_1} + \tfrac{(1-b_2)\lambda_2}{a_2} = 1 - b_2,
\end{align}
denoted as $d$, $e$, and $f$ in Fig. \ref{fig:anal_stableR_2TR}, respectively, $S_Q(\bm{\lambda})$ may have different shapes depending on how they intersect. Let $A= (x_A,y_A)$, $B= (x_B, y_B)$, and $C= (x_C, y_C)$ denote the intersection points of $e$ and $f$, $e$ and $d$, as well as $d$ and $f$, respectively, which can be obtained as 
\begin{align}
    A & = \left(\tfrac{a_1(1-b_2)^2}{b_1}, a_2 b_2 \right), \quad B  = \left(a_1b_1, \tfrac{a_2(1-b_1)^2}{b_2}\right), \nonumber \\
     C & = \left(a_1(1-b_2), a_2(1-b_1)\right).   
\end{align} 
Let us consider the following three cases:
\begin{enumerate}
    \item If $b_1 + b_2 > 1$, i.e., $0<\tfrac{\rho_{1,1}\rho_{2,2}}{\rho_{1,2}\rho_{2,1}\theta_1\theta_2}<1$, we have $x_A<x_C<x_B$ and $y_B<y_C<y_A$. In this case, $S_Q(\bm{\lambda})$ is illustrated in Fig. \ref{subfig:anal_stableR_AdHoc_PPPTR_2Txnum_case1_Pt-66_-69dBm_thres-5dB_-7dB_k_less_1}.
    \item If $b_1 + b_2 = 1$, i.e., $\tfrac{\rho_{1,1}\rho_{2,2}}{\rho_{1,2}\rho_{2,1}\theta_1\theta_2}=1$, we have $x_A = x_B = x_C$ and $y_A = y_B = y_C$. In this case, $S_Q(\bm{\lambda})$ is illustrated in Fig. \ref{subfig:anal_stableR_AdHoc_PPPTR_2Txnum_case1_Pt-66_-69dBm_thres-11.17dB_-7dB_k_eq_1}.
    \item If $0<b_1 + b_2<1$, i.e., $\tfrac{\rho_{1,1}\rho_{2,2}}{\rho_{1,2}\rho_{2,1}\theta_1\theta_2} > 1$, we have $x_B<x_C<x_A$ and $y_A<y_C<y_B$. In this case, $S_Q(\bm{\lambda})$ is illustrated in Fig. \ref{subfig:anal_stableR_AdHoc_PPPTR_2Txnum_case1_Pt-66_-69dBm_thres-11.17dB_-10dB_k_great_1}.
\end{enumerate}
By comparing Fig. \ref{subfig:anal_stableR_AdHoc_PPPTR_2Txnum_case1_Pt-66_-69dBm_thres-5dB_-7dB_k_less_1} with Fig. \ref{subfig:ga_stableR_AdHoc_PPPTR_2Txnum_case1_Pt-66_-69dBm_thres-5dB_-7dB_k_less_1}, we can see that Theorem \ref{theorem: SQ_2TR} and Algorithm \ref{alg:findSQ} lead to the same result. It can be further observed from Fig. \ref{fig:anal_stableR_2TR} that the stability region of input rates is significantly enlarged when the SINR threshold of any receiver reduces, in which case the receiver has a stronger capability of decoding packets from concurrent transmissions and is thus less affected by the interference from another transmitter.

\subsubsection{$K$ Symmetric T-R Pairs}\label{subsubsection: symmetric_stability region}
With $\lambda_i = \lambda$, $q_i = q$, $\theta_i = \theta$ for all $i\in\mathcal{K}$, and $\rho_{i,j} = \rho$ for all $i,j \in \mathcal{K}$, the following theorem presents the stability region $S_Q(\lambda)$ for $K$ symmetric T-R pairs. 
\begin{theorem}\label{theorem: SQ_symmetric}
    The stability region $S_Q(\lambda)$ for $K$ symmetric T-R pairs is 
    \begin{equation}\label{stability_region_symmetric}
        S_Q(\lambda) = \left\{ \lambda: 0 < \lambda < \lambda^u  \right\},
    \end{equation}
    where 
    \begin{equation}\label{lambda_u_symmetric}
        \lambda^u = \begin{cases}
            \tfrac{\theta + 1}{K\theta} \exp\left(-1-\tfrac{\theta}{\rho}\right)  & \text{if } \theta \geq \tfrac{1}{K - 1}, \\
            \exp\left( -\tfrac{K\theta}{\theta + 1}  -\tfrac{\theta}{\rho} \right)  & \text{otherwise}.
        \end{cases}
    \end{equation}
\end{theorem}
\begin{proof}
    See Appendix D.
\end{proof}

\begin{figure}[!t]
    \centering
    \includegraphics[width=2.3in,height=1.7in]{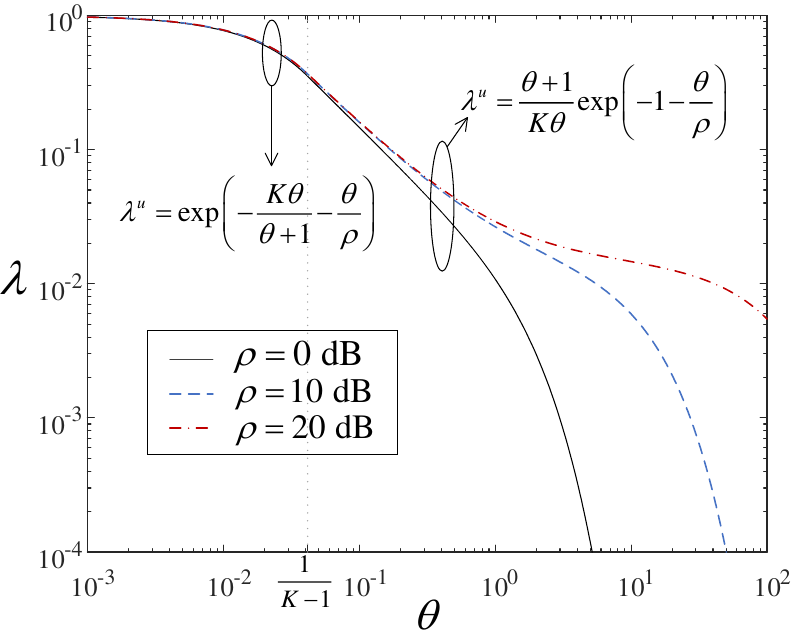}
    \vspace{-0.3cm}
    \caption{Stability region $S_{Q}(\lambda)$ obtained from Theorem \ref{theorem: SQ_symmetric} versus the SINR threshold $\theta$ for $K$ symmetric T-R pairs with the topology given in Fig. \ref{subfig:topology_singleCell_PPP_v5_K25}. $K = 25$.  }
    \label{fig:anal_stableR_symm_vs_theta_K25_v5_rho1_10_100dB}
    \vspace{-0.1cm}
\end{figure}

Note that the upper-bound $\lambda^u$ given in (\ref{lambda_u_symmetric}) is also the maximum throughput of a symmetric single-cell network in the saturated condition as given in Eq. (16) of \cite{MaximumsumrateCapture_Li}.
For illustration, we consider the symmetric single-cell network given in Fig. \ref{subfig:topology_singleCell_PPP_v5_K25}. Fig. \ref{fig:anal_stableR_symm_vs_theta_K25_v5_rho1_10_100dB} illustrates how the stability region $S_{Q}(\lambda)$ obtained from Theorem \ref{theorem: SQ_symmetric} changes with the SINR threshold $\theta$ for different values of mean received SNR $\rho$. By comparing Fig. \ref{fig:anal_stableR_symm_vs_theta_K25_v5_rho1_10_100dB} with Fig. \ref{subfig:exhaustiveSearch_stableR_symm_vs_theta_K25_v5_rho10dB}, it can be seen that Theorem \ref{theorem: SQ_symmetric} leads to the same result as the one obtained through Algorithm \ref{alg:findSQ}. Yet the effect of mean received SNR $\rho$ can be further observed from (\ref{lambda_u_symmetric}). As illustrated in Fig. \ref{fig:anal_stableR_symm_vs_theta_K25_v5_rho1_10_100dB}, the stability region $S_{Q}(\lambda)$ becomes smaller as $\rho$ decreases, indicating that the traffic input rate should be reduced accordingly for achieving stability for a smaller mean received SNR.   

\section{Simulation Results}\label{section: simulation}

In this section, simulation results will be presented to verify the preceding analysis. The simulation setting is the same as the system model described in Section \ref{section: model}, and each simulation is carried out for $10^8$ time slots. In the simulations, the steady-state probability of successful transmission of HOL packets of each transmitter is obtained by calculating the ratio of the number of successful transmissions to the total number of transmissions of HOL packets. The throughput of each transmitter is obtained by calculating the ratio of the number of successful transmissions to the number of time slots. 

\vspace{-0.3cm}
\subsection{Two T-R Pairs}\label{subsection: simu 2TR}

With $K = L = 2$, the network steady-state point $\bm{p}$ for given network state and the all-unsaturated region $S_{UU}(\bm{q}, \bm{\lambda})$ have been characterized in Section \ref{subsubsection:2_TR_p} and Theorem \ref{theorem: SUU_2TR}, respectively. 
To verify the analysis, Fig. \ref{subfig:simu_UU_p_PPPTR_Txnum2_case1_q2_0.7_Pt-66_-69dBm_thres-5dB_-7dB_irate0.2_0.27} and Fig. \ref{subfig:simu_UU_throughput_PPPTR_Txnum2_case1_q2_0.7_Pt-66_-69dBm_thres-5dB_-7dB_irate0.2_0.27} present the network steady-state point $\bm{p}$ and throughput $\bm{\lambda}_{out}$ for the two T-R pairs given in Fig. \ref{subfig:topology_AdHoc_PPPTR_Txnum2_case1}, respectively, with $q_1$ varying but $q_2$ fixed to $0.7$. According to the all-unsaturated region characterized in Fig. \ref{subfig:ga_UU_AdHoc_PPPTR_Txnum2_case1_Pt-66dBm_-69dBm_thres-5dB_-7dB_irate0.2_0.27} and Fig. \ref{subfig:SUU_AdHoc_PPPTR_Txnum2_case1_Pt-66_-69dBm_thres-5dB_-7dB_irate0.2_0.27}, with $q_2 = 0.7$ and $0.65< q_1 \leq 1$, the network would operate at the all-unsaturated steady-state point $\bm{p}_L^{K=2}$, where the throughput of each T-R pair $\lambda_{out, i} = \lambda_i$ for $i \in\{ 1, 2\}$. Otherwise, with $q_1 \leq 0.65$, Transmitter 2 would become saturated with its throughput dropping below the input rate, that is, the network is partially-saturated with $\bm{p} = \bm{p}_P^{K=2}$ and $\lambda_{out, 2} = \mu_2 < \lambda_2$.
The simulation results presented in Figs. \ref{subfig:simu_UU_p_PPPTR_Txnum2_case1_q2_0.7_Pt-66_-69dBm_thres-5dB_-7dB_irate0.2_0.27} and \ref{subfig:simu_UU_throughput_PPPTR_Txnum2_case1_q2_0.7_Pt-66_-69dBm_thres-5dB_-7dB_irate0.2_0.27} well agree with the analysis.

\begin{figure*}[!t]
    \centering
    \subfloat[]{
      \includegraphics[width=1.85in,height=1.38in]{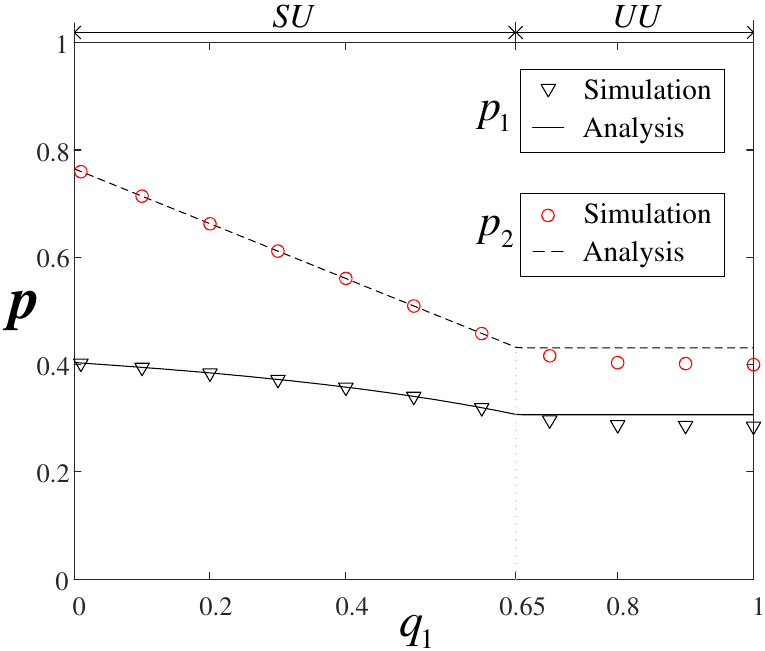}
    \label{subfig:simu_UU_p_PPPTR_Txnum2_case1_q2_0.7_Pt-66_-69dBm_thres-5dB_-7dB_irate0.2_0.27}}
    \hspace{0.3cm}
    \subfloat[]{
      \includegraphics[width=1.85in,height=1.38in]{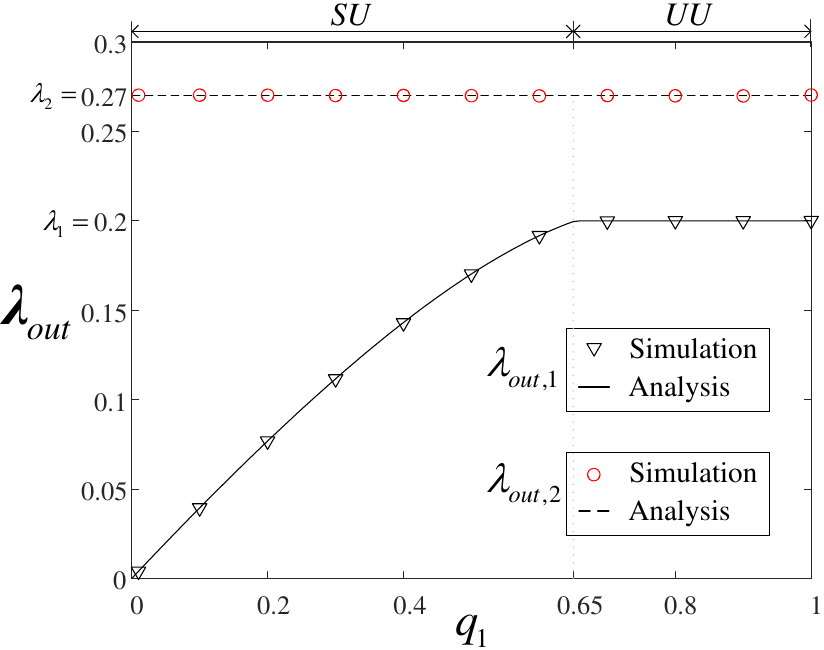}
    \label{subfig:simu_UU_throughput_PPPTR_Txnum2_case1_q2_0.7_Pt-66_-69dBm_thres-5dB_-7dB_irate0.2_0.27}}
    \hspace{0.3cm}
    \subfloat[]{
      \includegraphics[width=2.4in,height=1.38in]{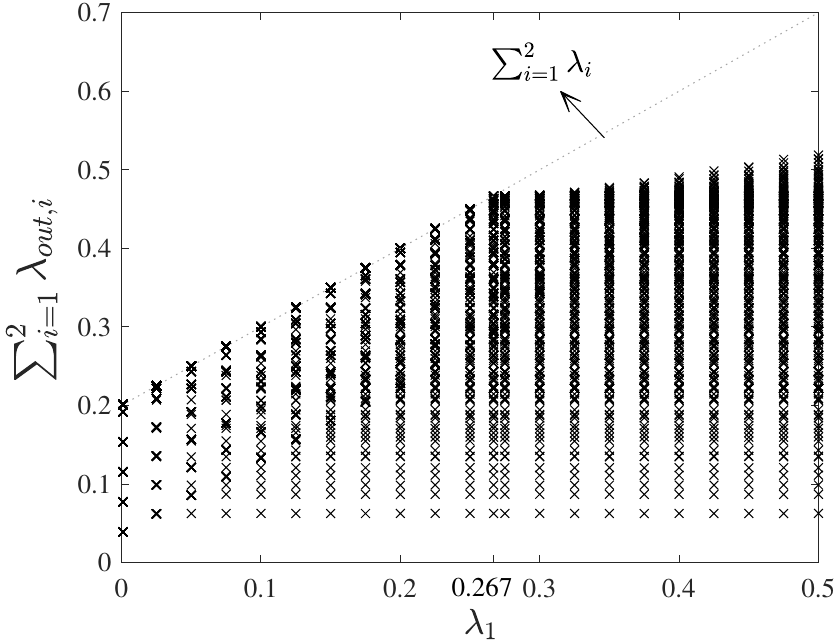}
    \label{subfig:simu_netthroughput_vs_irate1_q1q2_0.05to1_irate2_0.2_stableR_2TR_case1_thres5dB_7dB_Pt-66dBm_-69dBm}}
    \vspace{-0.2cm}
    \caption{(a) Network steady-state point $\bm{p}$ and (b) throughput $\bm{\lambda}_{out}$ versus transmission probability $q_1$ for two T-R pairs. $q_2 = 0.7$. $\lambda_1 = 0.2$. $\lambda_2 = 0.27$. (c) Simulated total throughput $\sum_{i = 1}^2 \lambda_{out, i}$ versus input rate $\lambda_1$. $q_1, q_2 \in \{0.05,0.1, \ldots, 0.95, 1\}$. $\lambda_2 = 0.2$. $\theta_1 = -5$ dB. $ \theta_2 = -7$ dB. }
    \label{fig:simu_UU_PPPTR_Txnum2_case1_case2_Pt-66_-69dBm_thres-5dB_-7dB_irate0.2_0.27}
    \vspace{-0.1cm}
\end{figure*}

To ensure that the all-unsaturated region is not empty, the input rates $\bm{\lambda}$ should be within the stability region $S_Q (\bm{\lambda})$, which has been characterized in Theorem \ref{theorem: SQ_2TR} for the case of two T-R pairs. To verify it, Fig. \ref{subfig:simu_netthroughput_vs_irate1_q1q2_0.05to1_irate2_0.2_stableR_2TR_case1_thres5dB_7dB_Pt-66dBm_-69dBm} presents the simulated total throughput of the two T-R pairs given in Fig. \ref{subfig:topology_AdHoc_PPPTR_Txnum2_case1}, where $\lambda_2 = 0.2$ and $\lambda_1$ varies from $0$ to $0.5$, with transmission probabilities 
$q_1$ and $q_2$ varying from $0.05$ to $1$ with an increment of $0.05$. According to the stability region $S_Q(\bm{\lambda})$ characterized in Fig. \ref{subfig:ga_stableR_AdHoc_PPPTR_2Txnum_case1_Pt-66_-69dBm_thres-5dB_-7dB_k_less_1} and Fig. \ref{subfig:anal_stableR_AdHoc_PPPTR_2Txnum_case1_Pt-66_-69dBm_thres-5dB_-7dB_k_less_1}, with $\lambda_2 = 0.2$, the network can be stabilized by properly choosing transmission probabilities $\bm{q}$ when $\lambda_1 < 0.267$. It can be observed from Fig. \ref{subfig:simu_netthroughput_vs_irate1_q1q2_0.05to1_irate2_0.2_stableR_2TR_case1_thres5dB_7dB_Pt-66dBm_-69dBm} that with $\lambda_1$ lower than $0.267$,  there exist values of $\bm{q}$ with which the total throughput of the two T-R pairs $\sum_{i=1}^2 \lambda_{out, i}$ is equal to their total input rate $\sum_{i=1}^2 \lambda_i$. In contrast, if $\lambda_1$ exceeds $0.267$, the network cannot be stabilized as the total throughput is always below the total input rate. 


\vspace{-0.17cm}
\subsection{$K$ Symmetric T-R Pairs}\label{subsection: simulation_SUU_K_TR}
\vspace{-0.03cm}

The network steady-state point $p$ and the all-unsaturated region $S_{U^K} (q, \lambda)$ for $K$ symmetric T-R pairs have been characterized in Section \ref{subsubsection: symm_K_TR_p} and Theorem \ref{theorem: SUU_symmetric}, respectively. To verify the analysis, Figs. \ref{subfig:simu_UU_p_q_symm_25Tx_v5_thres0dB_rho10dB_irate0.02} and \ref{subfig:simu_UU_netthroughput_q_symm_25Tx_v5_thres0dB_rho10dB_irate0.02}  show how the network steady-state point $p$ and total throughput $\sum_{i = 1}^K \lambda_{out, i}$ vary with the transmission probability $q$, respectively, for the symmetric single-cell network given in Fig. \ref{subfig:topology_singleCell_PPP_v5_K25}. According to the all-unsaturated region characterized in Fig. \ref{subfig:ga_SUU_symm_vs_irate_K25_v5_rho10dB_thres0dB} and Fig. \ref{fig:SUU_symm_vs_irate_K40_25_rho10dB_thres0dB}, for $\lambda = 0.02$ and $K = 25$, the all-unsaturated region can be obtained as $S_{U^K}(q, \lambda) = \{q: 0.034<q<0.157\}$. When $q$ is chosen from $S_{U^K}(q, \lambda)$, the network operates at the all-unsaturated steady-state point $p_L=0.608$ with the total throughput equal to the total input rate $K\lambda = 0.5$. Otherwise, with $q \leq 0.034$ or $q \geq 0.157$, the network becomes all-saturated, in which case the network steady-state point $p = p_A$ and the total throughput $\sum_{i = 1}^K \lambda_{out, i} = \sum_{i = 1}^K \mu_i < K\lambda$.

\begin{figure*}[!t]
    \centering
    \subfloat[]{
      \includegraphics[width=1.85in,height=1.4in]{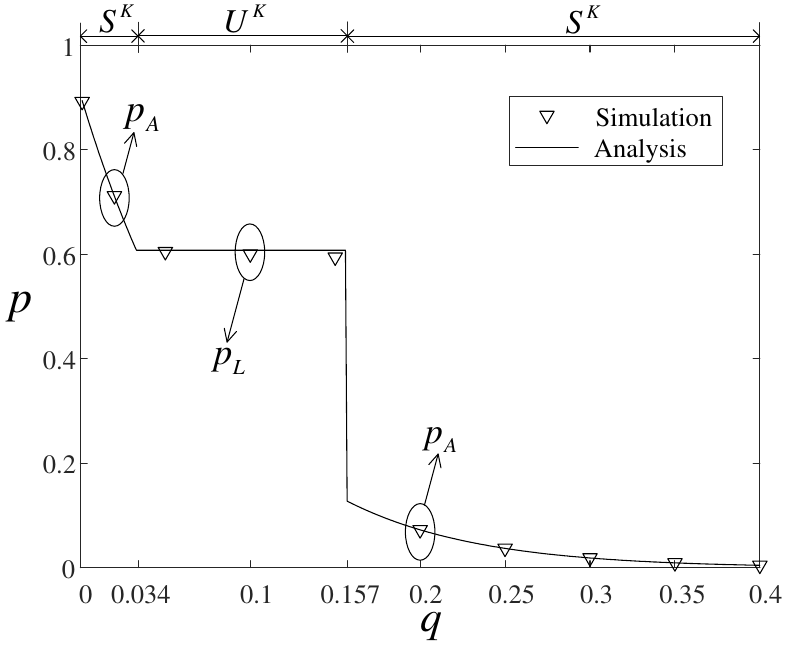}
    \label{subfig:simu_UU_p_q_symm_25Tx_v5_thres0dB_rho10dB_irate0.02}}
    \hspace{0.3cm}
    \subfloat[]{
      \includegraphics[width=1.85in,height=1.4in]{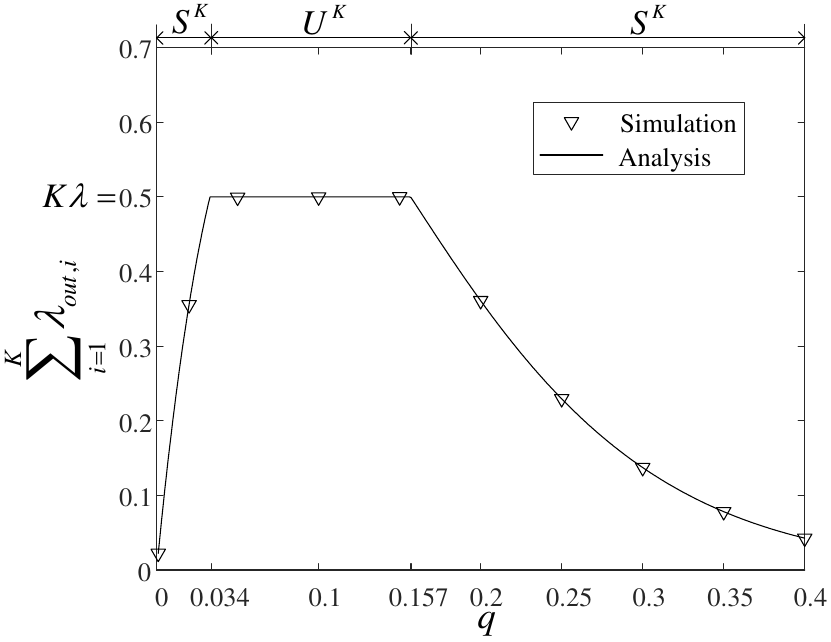}
    \label{subfig:simu_UU_netthroughput_q_symm_25Tx_v5_thres0dB_rho10dB_irate0.02}}
    \hspace{0.3cm}
    \subfloat[]{
        \includegraphics[width=2.4in,height=1.4in]{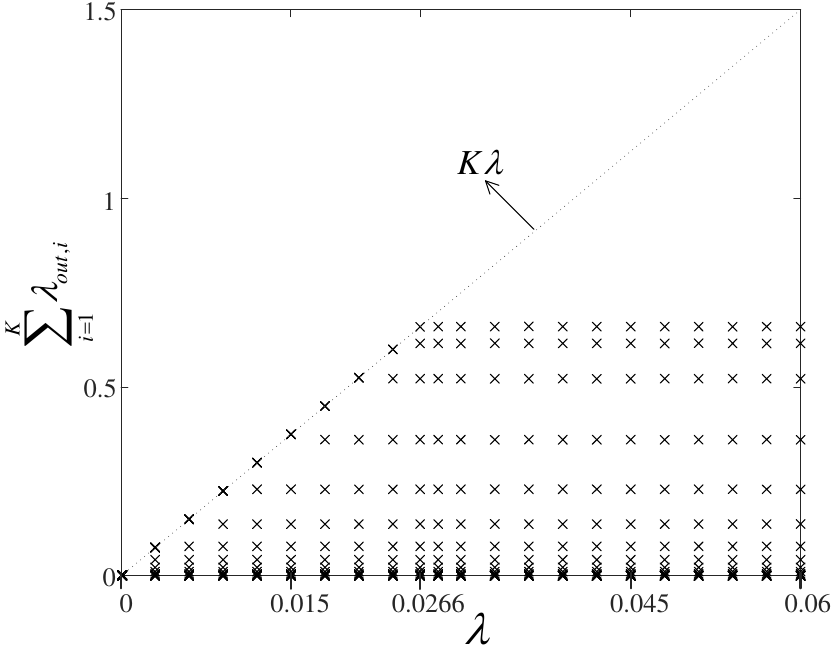}
      \label{subfig:simu_netthroughput_vs_rate0to0.06_changeq_stableR_symm_25Tx_v5_thres0dB_rho10dB}}
      \vspace{-0.2cm}
    \caption{(a) Network steady-state point $p$ and (b) total throughput $\sum_{i = 1}^K \lambda_{out, i}$ versus transmission probability $q$ for $K$ symmetric T-R pairs. $\lambda = 0.02$. (c) Simulated total throughput $\sum_{i = 1}^K \lambda_{out, i}$  versus input rate $\lambda$. $q \in \{0.05,0.1, \ldots, 0.95, 1\}$.  $K = 25$. $\rho = 10$ dB. $\theta= 0$ dB. }
    \label{fig:simu_UU_stableR_symm_25Tx_v5_thres0dB_rho10dB}
    \vspace{-0.4cm}
\end{figure*}

The stability region for the $K$ symmetric T-R pairs has been derived in Theorem \ref{theorem: SQ_symmetric}. To verify it, Fig. \ref{subfig:simu_netthroughput_vs_rate0to0.06_changeq_stableR_symm_25Tx_v5_thres0dB_rho10dB} shows how the simulated total throughput of the network changes with input rate $\lambda$ for transmission probability $q \in \{0.05,0.1, \ldots, 0.95, 1\}$. According to the stability region characterized in Fig. \ref{subfig:exhaustiveSearch_stableR_symm_vs_theta_K25_v5_rho10dB} and Fig. \ref{fig:anal_stableR_symm_vs_theta_K25_v5_rho1_10_100dB}, the network can be stabilized for $\lambda \in S_Q (\lambda) = \{\lambda: 0 < \lambda < \lambda^u = 0.0266 \}$. It can be observed from Fig. \ref{subfig:simu_netthroughput_vs_rate0to0.06_changeq_stableR_symm_25Tx_v5_thres0dB_rho10dB}  that with $\lambda < 0.0266$, the total network throughput is equal to the total input rate.
When $\lambda$ exceeds $0.0266$, the network throughput drops below the total input rate $K\lambda$. 

\vspace{-0.25cm}
\subsection{A Two-Cell Network}\label{subsection: simu two cell}
\vspace{-0.05cm}

\begin{figure}[!t]
    \centering
    \subfloat[]{
        \includegraphics[width=1.9in,height=1.5in]{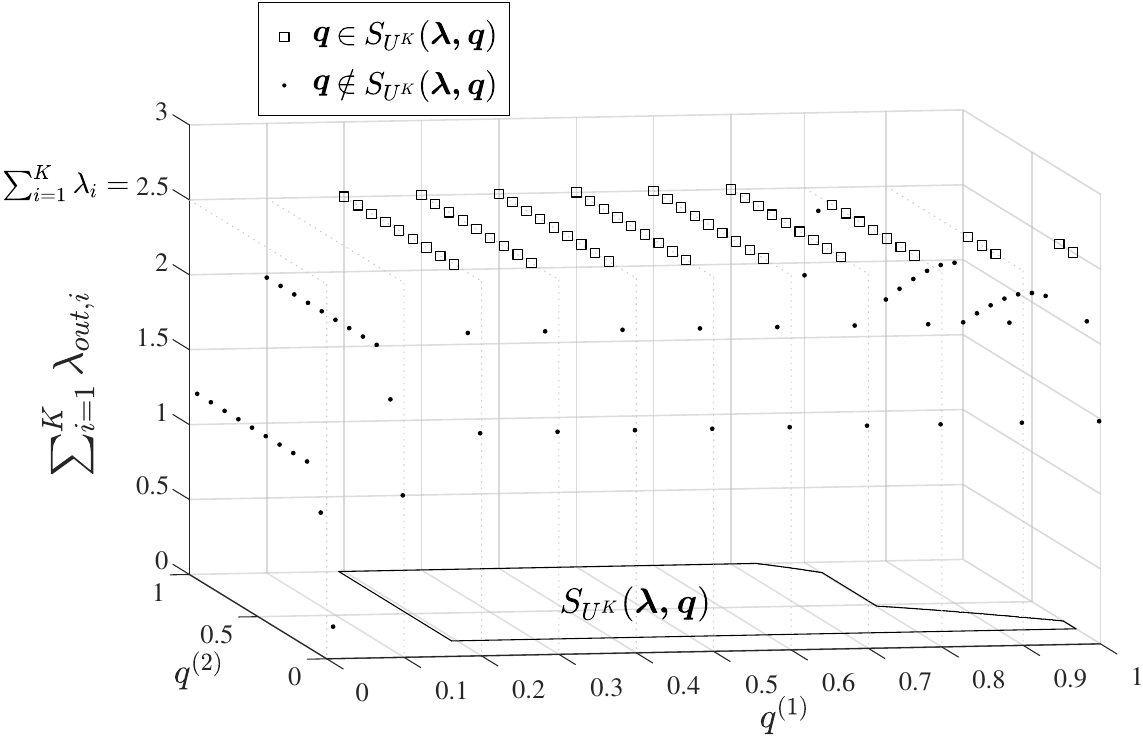}
      \label{subfig:simu_UU_2Cell_PPP_v3_rho0dB_irate0.1_thres-8dB}}
    \subfloat[]{
      \includegraphics[width=1.65in,height=1.5in]{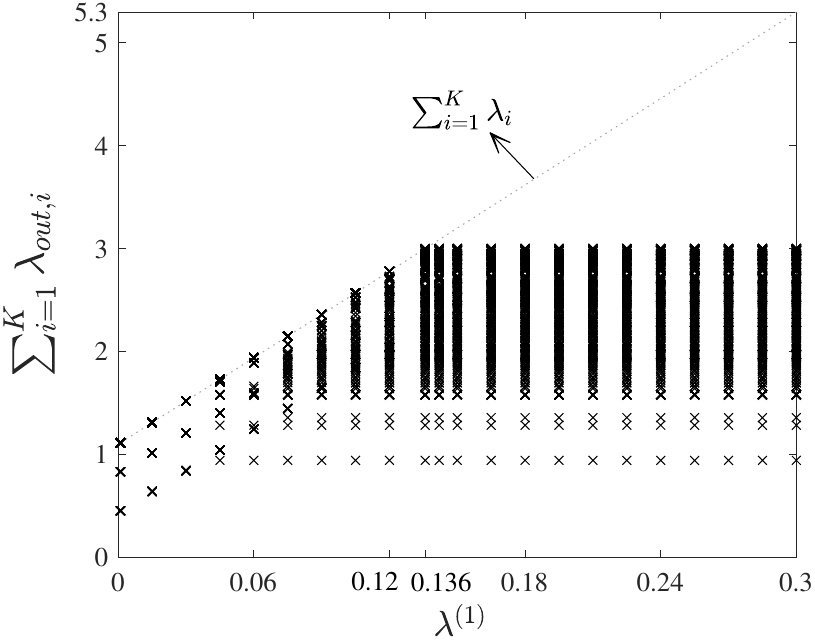}
    \label{subfig:simu_netthroughput_vs_irate1_q1q2_0.05to1_irate2_0.1_stableR_2Cell_PPP_v3_thres-8dB_rho0dB_new}}
    \vspace{-0.2cm}
    \caption{(a) Simulated total throughput $\sum_{i = 1}^K \lambda_{out, i}$ versus transmission probabilities $q^{(1)}$ and $q^{(2)}$ for a two-cell network.  $\lambda^{(1)} = \lambda^{(2)} = 0.1$.  (b) Simulated total throughput $\sum_{i = 1}^K \lambda_{out, i}$ versus input rate $\lambda^{(1)}$. $\lambda^{(2)} = 0.1$. $q^{(1)}, q^{(2)} \in \{0.05, 0.1, \ldots, 0.95, 1\}$. $\rho_{i,i^{\ast}} = 0$ dB, $i\in\mathcal{K}$. $\theta_1 = \theta_2 = -8$ dB.  $\alpha = 4$. }
    \label{fig:simu_UU_stableR_2Cell_PPP_v3_rho0dB_irate0.1_thres-8dB}
    \vspace{-0.1cm}
\end{figure}

In general, for a network with $K$ transmitters and $L$ receivers, the all-unsaturated region $S_{U^K}(\bm{q}, \bm{\lambda})$ and stability region $S_Q(\bm{\lambda})$ can be numerically obtained through Algorithm \ref{alg:findUU} and Algorithm \ref{alg:findSQ}, respectively. For a two-cell network with the topology given in Fig. \ref{subfig:topology_2Cell_PPP_v3}, for instance, the all-unsaturated region $S_{U^K}(\bm{q}, \bm{\lambda})$ for given input rates $\lambda^{(1)} = \lambda^{(2)} = 0.1$ has been illustrated in Fig. \ref{subfig:ga_UU_2Cell_PPP_v3_rho0dB_irate0.1_thres-8dB}. To verify it, simulation results presented in Fig. \ref{subfig:simu_UU_2Cell_PPP_v3_rho0dB_irate0.1_thres-8dB} show how the total throughput of the transmitters changes with the transmission probabilities $q^{(1)}$ and $q^{(2)}$. It can be observed that if $\bm{q} \in S_{U^K}(\bm{q}, \bm{\lambda})$, the total throughput is always equal to the total input rate $\sum_{i = 1}^{K} \lambda_i = \sum_{l = 1}^2 |\mathcal{K}^C_l| \lambda^{(l)} = 2.5$. In contrast, if $\bm{q} \notin S_{U^K}(\bm{q}, \bm{\lambda})$, the network throughput would drop below $2.5$, indicating that transmitters in at least one cell are saturated. For the stability region $S_Q(\bm{\lambda})$, according to Fig. \ref{subfig:ga_stableR_topology_2Cell_PPP_v3_thres-8dB_rho0dB}, for $\lambda^{(2)} = 0.1$, the input rate of each transmitter in Cell $1$, $\lambda^{(1)}$, should be smaller than $0.136$ for achieving stability, which can be verified by simulation results presented in Fig. \ref{subfig:simu_netthroughput_vs_irate1_q1q2_0.05to1_irate2_0.1_stableR_2Cell_PPP_v3_thres-8dB_rho0dB_new}.

\vspace{-0.2cm}
\subsection{An Ad-hoc Network}\label{subsection: simu ad-hoc}

For an ad-hoc network with multiple T-R pairs, in Fig. \ref{subfig:Adhoc_PPP_average9_v1}, transmitters are generated in $[0, 300\text{ m}]^2$ according to PPP with density $\xi = 10^{-4}$ m$^{-2}$ and each of them has a receiver with $d_{i,i} = 25$ m, $i\in\mathcal{K}$. For given input rate $\lambda_i = 0.2$, SINR threshold $\theta_i = 0$ dB, and transmission power $P_i = 17$ dBm for all $i\in\mathcal{K}$, the set of transmission probabilities for stabilizing the network can be obtained through Algorithm \ref{alg:findUU}, and Fig. \ref{subfig:Adhoc_PPP_average9_v1_with_stable_q} shows one of them. The simulation results presented in Fig. \ref{subfig:simu_stable_Adhoc_PPP_average9_v1_irate0.2_q_16throw_Tx17dBm_thres0dB_r25_Yang} demonstrate that the throughput of every T-R pair is equal to its input rate, indicating that the network can be stabilized with the transmission probabilities properly set.

\begin{figure*}[!t]
    \captionsetup[subfigure]{justification=centering}
    \centering
    \subfloat[]{
        \includegraphics[width=1.73in,height=1.63in]{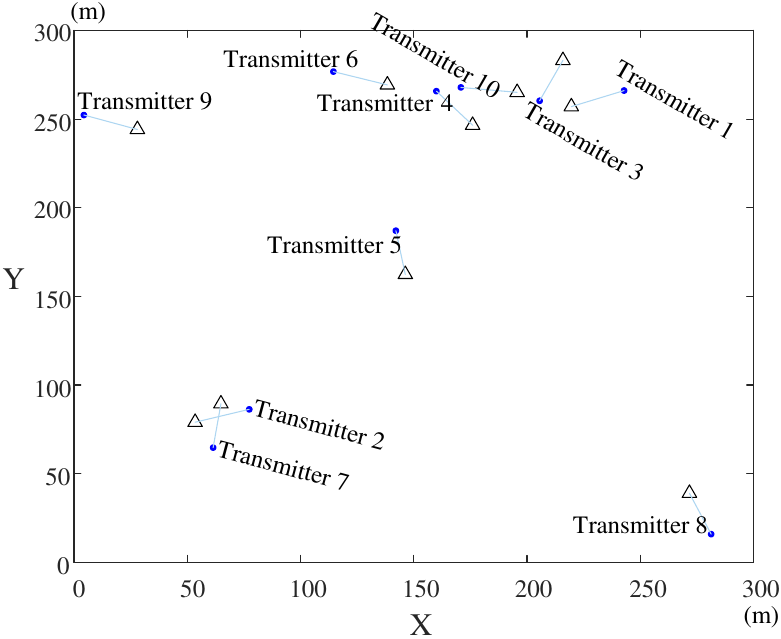}
    \label{subfig:Adhoc_PPP_average9_v1}}
    \subfloat[]{
        \includegraphics[width=1.73in,height=1.63in]{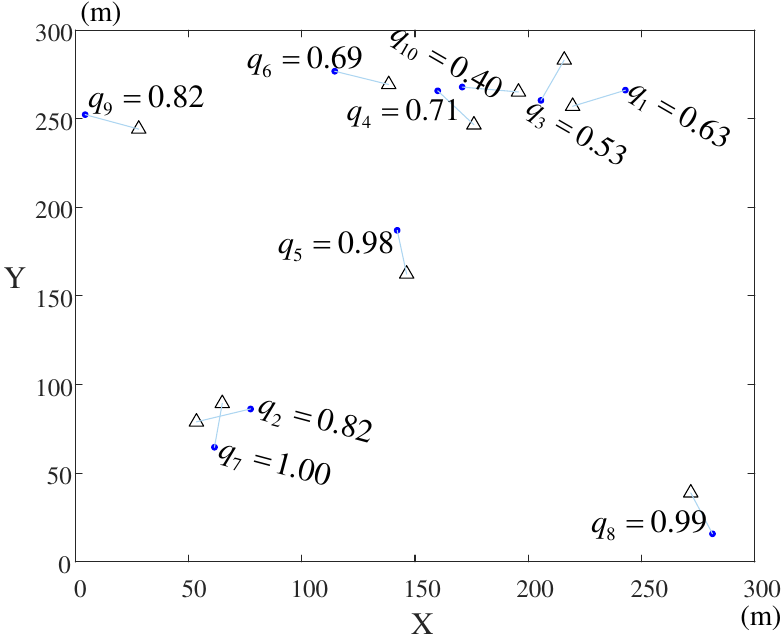}
    \label{subfig:Adhoc_PPP_average9_v1_with_stable_q}}
    \subfloat[]{
        \includegraphics[width=1.73in,height=1.63in]{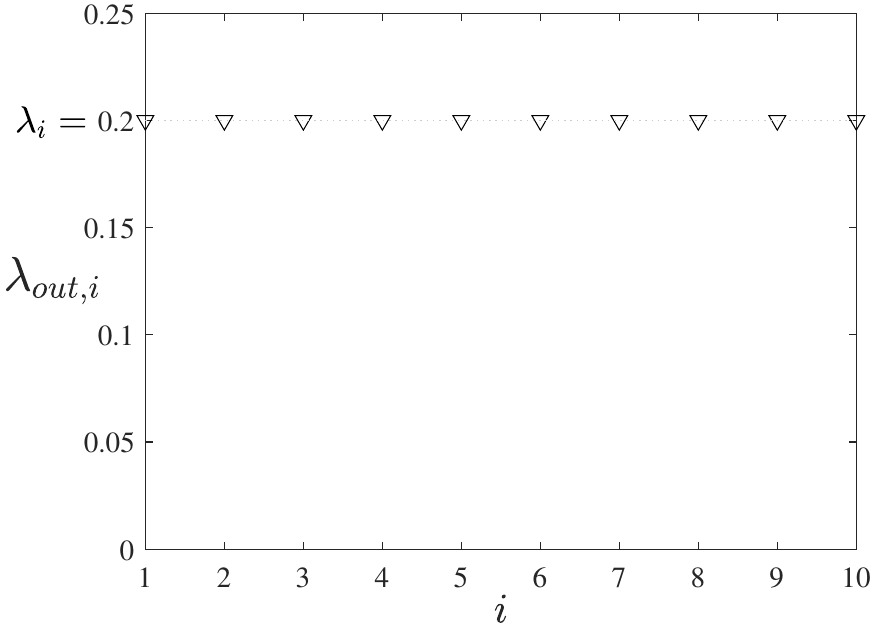}
    \label{subfig:simu_stable_Adhoc_PPP_average9_v1_irate0.2_q_16throw_Tx17dBm_thres0dB_r25_Yang}}
    \subfloat[]{
        \includegraphics[width=1.73in,height=1.63in]{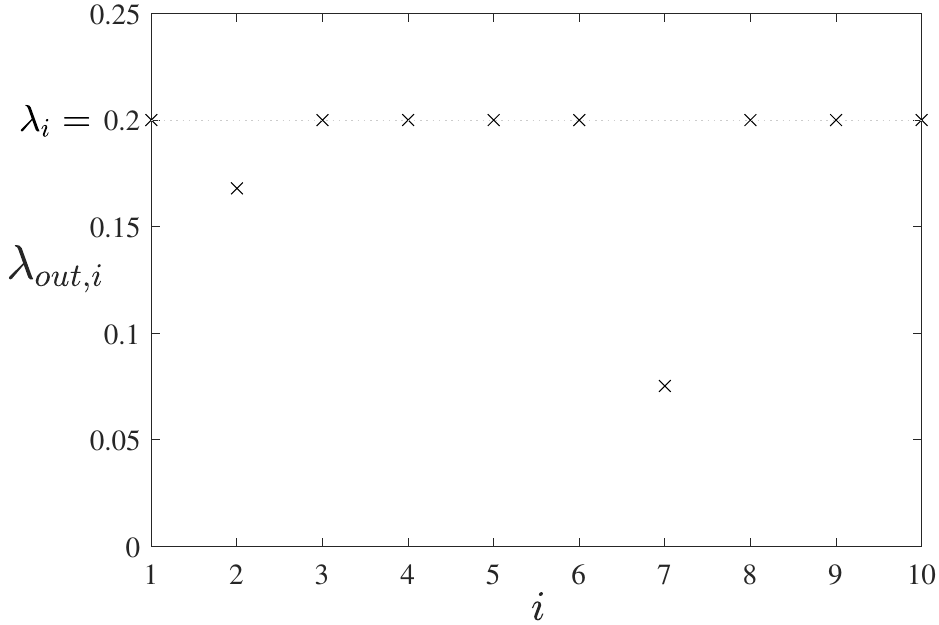}
    \label{subfig:simu_stable_Adhoc_PPP_average9_v1_irate0.2_q_all1_Tx17dBm_thres0dB_r25_Yang}}   
    \vspace{-0.1cm}
    \caption{(a) Network topology.  Transmitters and receivers are represented by dots and triangles, respectively. $\lambda_i = 0.2$, $\theta_i = 0$ dB, $P_i = 17$ dBm, $i \in \mathcal{K}$, $\sigma^2 =-90$ dBm, and $\alpha = 3.8$. (b) Transmission probabilities $\bm{q}$ obtained via Algorithm \ref{alg:findUU}.  (c)(d)  Simulated throughput $\lambda_{out, i}$ of T-R pair $i$, $i = 1,\ldots, 10$. (c) $\bm{q} = (0.63, 0.82, 0.53, 0.71, 0.98, 0.69, 1, 0.99, 0.82, 0.4)$, and (d) $\bm{q} = (1, 1, 1, 1, 1, 1, 1, 1, 1, 1)$.}
    \label{fig:Adhoc_PPP_average9_stable_q_irate0.2_Tx17dBm_thres0dB_r25_Yang}
    \vspace{-0.1cm}
\end{figure*}

\begin{figure*}[!t]
    \captionsetup[subfigure]{justification=centering}
    \centering
    \subfloat[]{
        \includegraphics[width=1.95in,height=1.68in]{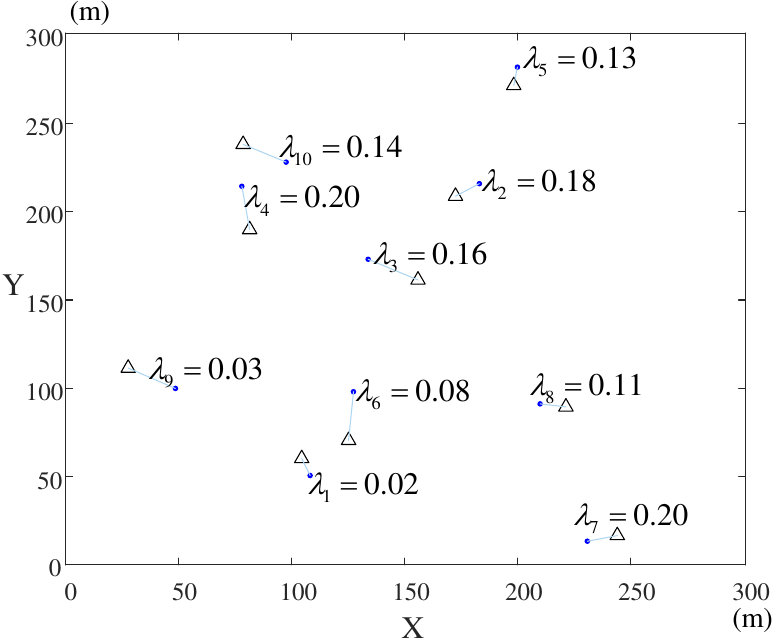}
    \label{subfig:topology_10Adhoc_dchange_v1}}
    \hspace{0.8cm}
    \subfloat[]{
        \includegraphics[width=1.95in,height=1.68in]{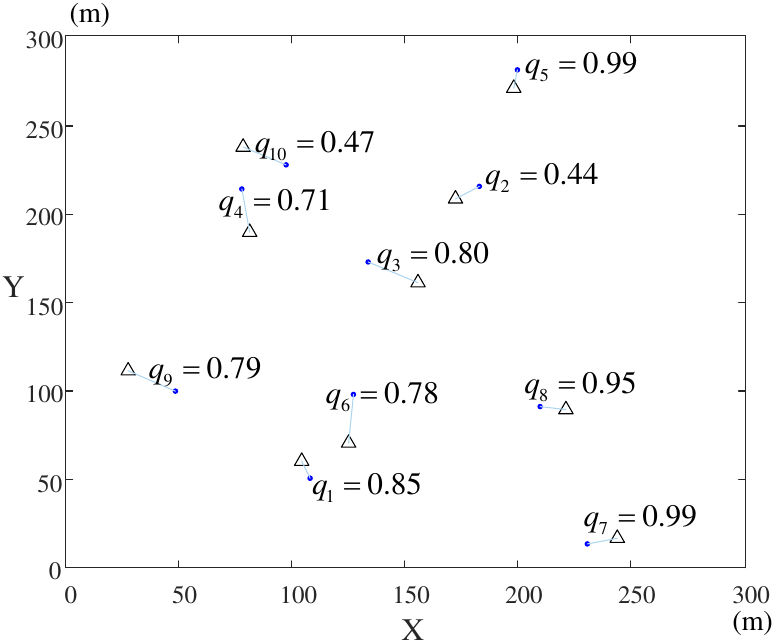}
    \label{subfig:topology_10Adhoc_dchange_v1_with_stable_q}}
    \hspace{0.75cm}
    \subfloat[]{
        \includegraphics[width=1.95in,height=1.68in]{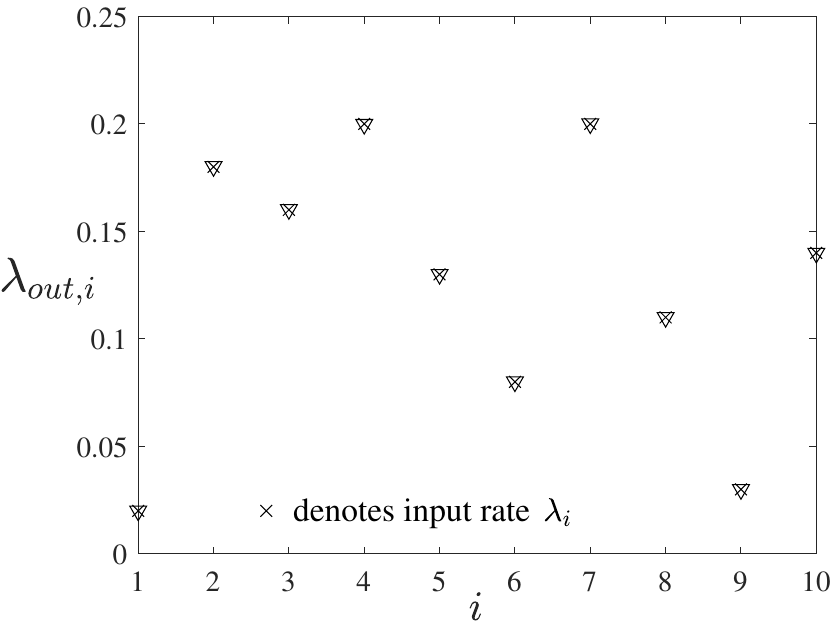}
    \label{subfig:simu_UU_10Adhoc_dchange_v1_ratechange_Pt17dBm_thres0dB_optq21th}}
    \vspace{-0.1cm}
    \caption{(a) Network topology with input rates $\bm{\lambda}$. Transmitters and receivers are represented by dots and triangles, respectively. $\theta_i=0$ dB, $P_i = 17$ dBm, $i\in\mathcal{K}$, $\sigma^2 =-90$ dBm, and $\alpha = 3.8$. (b) Transmission probabilities $\bm{q}$  obtained via Algorithm \ref{alg:findUU}. (c) Simulated throughput $\lambda_{out, i}$ of T-R pair $i$, $i = 1,\ldots, 10$, with the transmission probabilities given in (b). }
    \label{fig:topology_10Adhoc_dchange_with_simu}
    \vspace{-0.2cm}
\end{figure*}

Note that it was mentioned in Remark 1 that given the above setting, but the transmission probabilities are set to 1, then $90\%$ of the transmitters can be stabilized when the network area is sufficiently large. For an area of $[0, 300\text{ m}]^2$, nevertheless, the percentage of stable transmitters could be less than $90\%$. As Fig. \ref{subfig:simu_stable_Adhoc_PPP_average9_v1_irate0.2_q_all1_Tx17dBm_thres0dB_r25_Yang} illustrates, the throughputs of two T-R pairs drop below the input rate $0.2$, indicating that $2$ out of $10$ T-R pairs are unstable. As Fig. \ref{subfig:Adhoc_PPP_average9_v1_with_stable_q} illustrates, to stabilize the network, the transmission probability of each T-R pair should be carefully set based on the network topology. For those T-R pairs located in a dense area, transmission probabilities should be reduced accordingly to lower the interference. Adopting the same transmission probability across the network would lead to stronger interference for T-R pairs in the densely populated area and thus higher chances of being unstable.

In general, with an asymmetric setting where the T-R distance and input rate vary from T-R pair to T-R pair, e.g., the distance $d_{i,i}$ and input rate $\lambda_i$, $i\in\mathcal{K}$, are randomly generated from $10$ m to $30$ m and from $0$ to $0.2$, respectively, as shown in Fig. \ref{subfig:topology_10Adhoc_dchange_v1}, the set of transmission probabilities for stabilizing the network can still be obtained through Algorithm \ref{alg:findUU}, one of which is presented in Fig. \ref{subfig:topology_10Adhoc_dchange_v1_with_stable_q}. The simulation results shown in Fig. \ref{subfig:simu_UU_10Adhoc_dchange_v1_ratechange_Pt17dBm_thres0dB_optq21th} again verify that the network can be stabilized with the transmission probabilities properly set based on the input rates and network topology.

\section{Conclusion}\label{section: conclusion}

In this paper, the analytical framework proposed in \cite{Atheoreticalframework_Dai} is extended to incorporate multiple capture receivers for Aloha networks. The stability analysis shows that the network can be stabilized if and only if the input rates $\bm{\lambda}$ are within the stability region $S_{Q}(\bm{\lambda})$, and the transmission probabilities $\bm{q}$ are chosen from the all-unsaturated region $S_{U^K}(\bm{q}, \bm{\lambda})$, both of which are closely dependent on locations of transmitters and receivers, and can be characterized based on the exact steady-state probability of successful transmissions of HOL packets of each transmitter. The results are demonstrated in various examples, including the two T-R pairs and $K$ symmetric T-R pairs, where explicit expressions of $S_{U^K}(\bm{q}, \bm{\lambda})$ and $S_{Q}(\bm{\lambda})$ can be obtained, and the general multi-cell and ad-hoc scenarios, where $S_{U^K}(\bm{q}, \bm{\lambda})$ and $S_{Q}(\bm{\lambda})$ can only be numerically calculated. 
Simulation results validate the analysis and corroborate that adopting a fixed and identical transmission probability across the network is the root cause of failing to achieve network stability. Instead, to stabilize the whole network, transmission probabilities of T-R pairs should be carefully adjusted based on their locations and traffic input rates.


\begin{appendices}
    \vspace{-0.1cm}
\section{Proof of Theorem \ref{theorem: SUU_2TR}} \label{appendix: SUU_K=2}
\vspace{-0.1cm}

With $K=L=2$, for the network to be all-unsaturated, the transmission probabilities $\bm{q}$ should satisfy certain constraints to ensure: 1) the convergence to the all-unsaturated steady-state point $\bm{p}_{L}^{K=2}$, and 2) the service rate of each transmitter's queue being larger than its input rate at $\bm{p}_{L}^{K=2}$.

For the convergence to $\bm{p}_{L}^{K=2}$, consider the iterative process $c_{t+1} = F(c_t)$, where $F(c_t) = \prod_{i=1}^{2} \frac{a_i c_t}{a_i c_t + b_i \lambda_i}$, $t = 0, 1, \ldots$. $c_t$ needs to converge to the attracting fixed point $c_L$ of (\ref{c_2_K=2}). The following lemma presents a sufficient and necessary condition for the convergence to $c_L$.  
\begin{lemma}\label{lemma: ct cL}
    A sufficient and necessary condition for the convergence of $c_t$ to $c_L$ is $\min_t c_t > c_S$.
\end{lemma}
\begin{proof}
    It can be easily obtained that $0<F'(c_L) < 1$ and $F'(c_S) >1$, where $F'(c)$ is the first-order derivation of $F(c)$. As a result, when $c_t < c_S$, $c_{t+1} = F(c_t) < c_t$. In that case, $c_t$ monotonically decreases and diverges from $c_S$. On the other hand, if $c_t > c_S$, we have $c_{t+1} = F(c_t) > c_t$ when $c_t < c_L$, and $c_{t+1} = F(c_t) < c_t$ when $c_t > c_L$, indicating that $c_t$ converges to $c_L$. We can see that a sufficient and necessary condition for $c_t$ to converge to $c_L$ is $c_t > c_S$ for all $t$. 
\end{proof}
By combining Lemma \ref{lemma: ct cL} with (\ref{c_K=2}), we can obtain that the convergence to $\bm{p}_{L}^{K=2}$ requires  
\begin{equation}\label{K=2_convergence}
    \left(1 - \tfrac{b_1 \lambda_1}{p_{1,t}} \right)\left(1 - \tfrac{b_2 \lambda_2}{p_{2,t}} \right) > \left(1 - \tfrac{b_1 \lambda_1}{p_{1,S}^{K=2}} \right)\left(1 - \tfrac{b_2 \lambda_2}{p_{2,S}^{K=2}} \right), \; \text{for all }t,
\end{equation}
where $p_{i,t} = a_i c_t + b_i \lambda_i$. 
Lemma \ref{lemma: condition_convergence_K=2} further shows the sufficient and necessary condition for (\ref{K=2_convergence}).
\begin{lemma}\label{lemma: condition_convergence_K=2}
    A sufficient and necessary condition for (\ref{K=2_convergence}) is 
    \begin{equation}
        q_1 < \tfrac{\lambda_1}{p_{1,S}^{K=2}}, \quad \text{or} \quad q_2 < \tfrac{\lambda_2}{p_{2,S}^{K=2}}.
    \end{equation}
\end{lemma}
\begin{proof}
    \textit{Sufficient Condition:} If $q_i < \tfrac{\lambda_i}{p_{i,S}^{K=2}}$, $i\in\{1,2\}$, then for Transmitter $j\neq i$,
    \begin{equation}
        p_{j,t} \geq a_j (1 - b_i q_i) > a_j (1 - \tfrac{b_i\lambda_i}{p_{i,S}^{K=2}}) = p_{j,S}^{K=2}, \quad \text{for all } t.
    \end{equation}
    For Transmitter $i$,
    \begin{equation}
        p_{i,t} = a_i\left(1 - \tfrac{b_j \lambda_j}{p_{j,t-1}}\right) > a_i\left(1 - \tfrac{b_j \lambda_j}{p_{j,S}^{K=2}}\right) = p_{i,S}^{K=2}, \quad \text{for all } t.
    \end{equation}
    With both $p_{i,t} > p_{i,S}^{K=2}$ and $p_{j,t} > p_{j,S}^{K=2}$ for all $t$, (\ref{K=2_convergence}) holds true.

    \textit{Necessary Condition: } If $q_i \geq \tfrac{\lambda_i}{p_{i,S}^{K=2}}$ for both $i =1$ and $i =2$, then in the worst scenario when both transmitters' queues are busy, we have
    \begin{equation}
        p_{i,t} = a_i\left(1 - b_j q_j \right) \leq a_i\left(1 - \tfrac{b_j \lambda_j}{p_{j,S}^{K=2}}\right) = p_{i,S}^{K=2},
    \end{equation}
    for both $i = 1$ and $i = 2$, with which (\ref{K=2_convergence}) does not hold.
\end{proof}

To ensure that the service rate of each transmitter is larger than its input rate, according to (\ref{service_rate}), we have 
\begin{equation}\label{SUU_2TR_lowbound}
    q_i > \tfrac{\lambda_i}{p_{i,L}^{K=2}},
\end{equation}
for both $i = 1$ and $i=2$.

Finally, Theorem \ref{theorem: SUU_2TR} can be obtained by combining Lemma \ref{lemma: condition_convergence_K=2}, Lemma \ref{lemma: ct cL} and (\ref{SUU_2TR_lowbound}).

\vspace{-0.2cm}
\section{Proof of Theorem \ref{theorem: SUU_symmetric}}\label{appendix: SUU_symmetric}
\vspace{-0.1cm}

With $K$ symmetric T-R pairs, for the network to be all-unsaturated, the transmission probability $q$ should satisfy certain constraints to ensure: 1) the convergence to the all-unsaturated steady-state point $p_L$, and 2) the service rate of each transmitter's queue being larger than its input rate at $p_L$.

For the convergence to $p_{L}$, consider the iterative process $p_{t+1} = G(p_t)$, where $G(p_t) = \exp\left( -\tfrac{\theta}{\rho} - \tfrac{K \theta}{\theta + 1}\cdot \tfrac{\lambda}{p_t} \right)$, $t = 0, 1, \ldots$. $p_t$ needs to converge to the attracting fixed point $p_L$ of (\ref{p_single_symmetric_unsaturated}). By following a similar derivation to Lemma \ref{lemma: ct cL}, a sufficient and necessary condition for the convergence to $p_L$ can be obtained as $\min_t p_t > p_S$. Lemma \ref{lemma: condition_convergence_symmetric} further
shows the sufficient and necessary condition for $\min_t p_t > p_S$.
\begin{lemma}\label{lemma: condition_convergence_symmetric}
    A sufficient and necessary condition for $\min_t p_t > p_S$ is $q < \tfrac{\lambda}{p_S}$.
\end{lemma}
\begin{proof}
    \textit{Sufficient Condition: } If $q < \tfrac{\lambda}{p_S}$, then we have 
    \begin{equation}
        p_t \geq \exp \left( -\tfrac{\theta}{\rho} - \tfrac{K\theta}{\theta + 1} \cdot q \right) > \exp \left( -\tfrac{\theta}{\rho} - \tfrac{K\theta}{\theta + 1} \cdot \tfrac{\lambda}{p_S} \right) = p_S, \text{for all }t.
    \end{equation}
    with which $\min_t p_t > p_S$ holds true.

    \textit{Necessary Condition: } If $q \geq \tfrac{\lambda}{p_S}$, then in the worst scenario when all transmitters are busy, we have
    \begin{equation}
        p_t = \exp \left( -\tfrac{\theta}{\rho} - \tfrac{K\theta}{\theta + 1} \cdot q \right) \leq \exp \left( -\tfrac{\theta}{\rho} - \tfrac{K\theta}{\theta + 1} \cdot \tfrac{\lambda}{p_S} \right) = p_S,
    \end{equation}
    with which $\min_t p_t > p_S$ does not hold.
\end{proof}

To ensure that the service rate of each transmitter is larger than its input rate, according to (\ref{service_rate}), we have 
\begin{equation}\label{SUU_symmetric_lowbound}
    q > \tfrac{\lambda}{p_{L}}.
\end{equation}

Finally, Theorem \ref{theorem: SUU_symmetric} can be obtained by combining Lemma \ref{lemma: condition_convergence_symmetric} with (\ref{SUU_symmetric_lowbound}).

\vspace{-0.1cm}
\section{Proof of Theorem \ref{theorem: SQ_2TR}}\label{appendix: SQ_2TR}
\vspace{-0.1cm}
With $K=L=2$, the network can be stabilized only when the all-unsaturated steady-state point $\bm{p}_L^{K=2}$ exists, i.e., (\ref{condition_1}) holds. Furthermore, to ensure that there exist transmission probabilities $\bm{q}$ to achieve stability, i.e., $\bm{\lambda}<\bm{\mu}$, according to (\ref{service_rate}), it also requires $\lambda_i<p_{i,L}^{K=2} \cdot q_i \leq p_{i,L}^{K=2}$, which can be solved as 
\begin{equation}\label{condition_2}
    \tfrac{(2-b_i)\lambda_i}{a_i} + \tfrac{b_j\lambda_j}{a_j} < 1,\!  \text{ or } \! \tfrac{(1-b_i)\lambda_i}{a_i} + \tfrac{b_j \lambda_j}{a_j}\! <\! 1-b_i,  i, j \in \{1, 2\}, i\!\neq\! j.
\end{equation}
Theorem \ref{theorem: SQ_2TR} can be obtained by combining (\ref{condition_1}) and (\ref{condition_2}).

\vspace{-0.2cm}
\section{Proof of Theorem \ref{theorem: SQ_symmetric}}\label{appendix: SQ_symmetric}

With $K$ symmetric T-R pairs, the network can be stabilized only when the all-unsaturated steady-state point $p_L$ exists, i.e., (\ref{symm_condition_1}) holds. Furthermore, to ensure that there exists transmission probability $q$ to achieve stability, it also requires $\lambda <p_{L}  \cdot q \leq p_{L}$, which can be solved as 
\begin{equation}\label{symm_condition_2}
    0<\lambda< \exp\left( -\tfrac{K\theta}{\theta + 1}  -\tfrac{\theta}{\rho} \right), \quad \text{or} \quad \theta > \tfrac{1}{K-1}. 
\end{equation}
Theorem \ref{theorem: SQ_symmetric} can be obtained by combining (\ref{symm_condition_1}) and (\ref{symm_condition_2}).

\end{appendices}

\balance

\end{document}